\documentclass[a4paper,english]{article}
\usepackage[utf8]{inputenc}
\usepackage{lmodern}
\usepackage[T1]{fontenc} 
\usepackage{babel}

\usepackage{amsmath}
\usepackage{IEEEtrantools}
\usepackage{mathtools}
\mathtoolsset{mathic=true}
\usepackage{amsfonts,amssymb}
\usepackage{dsfont}
\usepackage{mleftright}

\usepackage{amsthm}
\newtheorem{definition}{Definition}

\usepackage{fullpage}
\usepackage{enumitem}
\usepackage[affil-it]{authblk}

\usepackage[sort&compress,numbers,merge]{natbib}
\usepackage[hidelinks,unicode]{hyperref}

\DeclareMathOperator{\tr}{Tr}
\DeclareMathOperator{\interior}{int}

\DeclareMathOperator*{\maximize}{maximize}
\DeclareMathOperator*{\minimize}{minimize}

\newcommand{\bvec}{\mathbf}

\newcommand{\id}{\mathds{1}}
\newcommand{\rchsh}{\mathrm{chsh}}

\usepackage{braket}

\let\lsbkt[ \let\rsbkt]
\let\lparn( \let\rparn)
\DeclarePairedDelimiter{\cointerval}{\lsbkt}{\lsbkt}

\DeclarePairedDelimiter{\ccinterval}{\lsbkt}{\rsbkt}

\DeclarePairedDelimiter{\abs}{\lvert}{\rvert}
\DeclarePairedDelimiter{\mean}{\langle}{\rangle}
\DeclarePairedDelimiter{\iprod}{\langle}{\rangle}

\newcommand{\given}[1][]{\nonscript\:#1\vert\nonscript\:}

\newtheorem{theorem}{Theorem}

\newtheorem{lemma}{Lemma}

\newcommand{\bff}{\bvec{f}}
\newcommand{\freg}{\ccinterval{\hat{\bff}^-, \hat{\bff}^+}}

\newcommand{\tresh}{H_{\mathrm{thr}}}

\begin{document}
\title{Device-independent randomness generation\\ from several Bell estimators}
\author{Olmo Nieto-Silleras, C\'edric Bamps, Jonathan Silman, Stefano Pironio}
\affil{Laboratoire d'Information Quantique, CP224\\ Universit\'e libre de Bruxelles, 1050 Brussels (Belgium)}
\date{(Dated: March 20, 2018)}

\maketitle

\begin{abstract}
Device-independent randomness generation and quantum key distribution protocols rely on a fundamental relation between the non-locality of quantum theory and its random character.
This relation is usually expressed in terms of a trade-off between the probability of guessing correctly the outcomes of measurements performed on quantum systems and the amount of violation of a given Bell inequality.
However, a more accurate assessment of the randomness produced in Bell experiments can be obtained if the value of several Bell expressions is simultaneously taken into account, or if the full set of probabilities characterizing the behavior of the device is considered.
We introduce protocols for device-independent randomness generation, secure against classical side information, that rely on the estimation of an arbitrary number of Bell expressions or even directly on the experimental frequencies of measurement outcomes.
Asymptotically, this results in an optimal generation of randomness from experimental data (as measured by the min-entropy), without having to assume beforehand that the devices violate a specific Bell inequality.
\end{abstract}

\section{Introduction}
In recent years, researchers have uncovered a fundamental relationship between the non-locality of quantum theory and its random character.
This relationship is usually formulated as follows.
Consider two (or generally $k$) separated quantum devices accepting, respectively, classical inputs $\mathrm{x}_1$ and $\mathrm{x}_2$ and outputting classical outputs $\mathrm{a}_1$ and $\mathrm{a}_2$.
Let $p=\{p(\mathrm{a}_1\mathrm{a}_2 \given \mathrm{x}_1\mathrm{x}_2)\}$ denote the set of joint probabilities describing how the devices respond to given inputs, from the point of view of a user who can only interact with the devices through the input-output interface, but who has no knowledge of the inner workings of the devices.
Suppose that given $p$, the expectation value of a certain Bell expression $f$, such as the Clauser-Horne-Shimony-Holt (CHSH) expression \cite{CHSH69}, equals $f[p]$.
Then, it is in principle possible to compute a lower bound on the randomness generated by the devices, as quantified by the min-entropy---the negative logarithm of the maximal probability of correctly guessing the values of future outputs.
This bound on the min-entropy holds for any observer, including those having an arbitrarily precise description of the inner workings of the devices, and depends only on information derived from the resulting input-output behavior through the quantity $f[p]$.
In principle, this bound can be computed numerically for any given Bell expression $f$.
For certain Bell expressions, such as the CHSH expression, it can also be determined analytically.

This relation between the non-locality of quantum theory and its randomness is at the basis of various protocols for device-independent (DI) randomness generation (RNG) \cite{Col09,*CK11,PAM+10} and quantum key distribution (QKD) \cite{MY98,ABG+07}.
The theoretical analysis of such protocols presents us with an extra challenge in that the probabilistic behavior $p$ of the devices is not known in advance and may vary from one measurement run to the next.
This implies that bounds on the randomness as a function of $f[p]$ have to be adapted to rely instead on the value of the Bell expression $f$ estimated from experimental data.
Some DIRNG and DIQKD protocols, and their security analyses, are reliant on specific Bell inequalities (usually the CHSH inequality) \cite{VV12,RUV13,CY13,ARV16} or certain families of Bell inequalities \cite{Mas09,MS14,MRC+14}, while others may be adapted to arbitrary Bell inequalities \cite{PAM+10,MPA11,PM13,FGS13,PMLA13,MS14a}.
However, to our knowledge all DIRNG and DIQKD protocols in the literature require that a single Bell inequality be chosen in advance and its experimental violation estimated (one exception is \cite{RBH+16}, where two fixed Bell expressions are used).
The length and secrecy of the final key will then depend on the observed violation of the chosen inequality.

Nevertheless, it has been pointed out in \cite{NPS14,BSS14} that the fundamental relation between the randomness and non-locality of quantum theory does not necessarily need to be expressed in terms of a specific Bell inequality.
It is in principle possible, at least numerically, to bound the probability of guessing correctly the outputs of a pair of quantum devices directly from the knowledge of the joint input-output probabilities $p$.
Indeed, the amount of violation $f[p]$ of a given Bell inequality captures the non-local behavior of the devices only partially, and better bounds on the min-entropy can be obtained if all the information about the devices' behavior is taken into account.

This observation raises the following question: can one devise a device-independent RNG or QKD protocol that does not rely on the estimation of any a priori chosen Bell inequality, but which instead takes directly into account all the data generated by the devices?

There are various reasons for introducing protocols of this type.
First, as already mentioned, the entire set of data generated by the devices can provide more information than the violation of a specific Bell inequality, and may therefore potentially allow for more efficient protocols.
Second, the choice of a Bell inequality may have a deep influence on the amount of randomness that can be certified: as shown in \cite{AMP12} there are devices for which the amount of randomness, as computed from the CHSH inequality, is arbitrarily small, but is maximal if computed using another Bell inequality.
Third, even if a set of quantum devices have been specifically designed to maximize the randomness according to a specific Bell inequality, the optimal extraction of randomness from noisy versions of such devices, say because of degradation of the devices with time, will typically rely on other Bell inequalities \cite{NPS14,BSS14,MP13}.
Finally, suppose that one is given a set of quantum devices without any specification of which Bell inequality they are expected to violate.
Can one nevertheless directly use them in a protocol and obtain a non-zero random string or shared key, without testing their behavior beforehand?

We show here that it is indeed possible to devise DIRNG protocols which exploit more information than the estimated violation of a single Bell inequality, particularly, DIRNG protocols which exploit the full set of frequencies obtained (i.e., the entire set of estimates of the behavior $p$).
Specifically, we introduce a DIRNG protocol whose security holds against an adversary limited to classical side information, or equivalently, with no long-term quantum memory.
(Note that such a level of security may well be sufficient for all practical purposes \cite{PM13,FGS13}.)
Technically, our protocol is obtained by generalizing the security analysis introduced in \cite{PM13,FGS13} and combining it with the semidefinite programming techniques introduced in \cite{NPS14,BSS14} for lower-bounding the randomness based on the full set of probabilities $p$ (which cannot be directly applied to experimental data).

We start in Section~\ref{sec:bell} by briefly presenting  the theoretical framework of our work, its main assumptions, and the notation used throughout the paper.
In Sections~\ref{sec:rand} to \ref{sec:sdp} we present our main mathematical results.
In Section~\ref{sec:rand} we present the main theorem of the paper and explain in detail how to put a DI bound on the randomness produced when measuring a Bell device $n$ times in succession, given that we have a way to bound the single-round randomness as a function of the Bell expectation, and given that we can estimate the Bell expectation with some confidence.
These two sub-procedures are respectively presented in Sections~\ref{sec:estim} and \ref{sec:sdp} for the general case of an arbitrary number of Bell expressions.
Combining these two sub-procedures with the general approach of Section~\ref{sec:rand} immediately yields a DIRNG protocol, whose various steps are summarized in Section~\ref{sec:prot}.
In Section~\ref{sec:discussion} we discuss in detail the main features of our protocol, and illustrate these with a numerical example.
We end with some concluding remarks and open questions in Section~\ref{sec:conclusion}.

\section{Behaviors and Bell expressions}\label{sec:bell}

In the following we will refer to a Bell set up, that is to say, $k$ separated ``black'' boxes (quantum devices whose inner workings are unknown), as a Bell device.
Each box $i$ can receive an input $\mathrm{x}_i$ upon which it produces an output $\mathrm{a}_i$, with $\mathrm{x}_i$ and $\mathrm{a}_i$ taking values in some finite sets $\mathcal{X}_i$ and $\mathcal{A}_i$, respectively, where without loss of generality we assume that the set of outputs $\mathcal{A}_i$ does not depend on the input $\mathrm{x}_i$.
We write $x=(\mathrm{x}_1,\dotsc,\mathrm{x}_k)$ and $a=(\mathrm{a}_1,\dotsc,\mathrm{a}_k)$ for the $k$-tuple of inputs and outputs, and write $\mathcal{X}=\mathcal{X}_1\times\dotsm\times\mathcal{X}_k$ and $\mathcal{A}=\mathcal{A}_1\times\dotsm\times\mathcal{A}_k$ for the set of all possible $k$-tuples of inputs and outputs.
Note that we use a roman (upright) type for the inputs and outputs of a single box and an italic type for the joint inputs and outputs of all $k$ boxes.

The behavior of a single-round use of this Bell device can be characterized by the $\abs{\mathcal A} \times \abs{\mathcal X}$ joint probabilities $p(a \given x)$, which we can arrange into a vector $p \in \mathbb{R}^{\abs{\mathcal A} \times \abs{\mathcal X}}$.
We denote by $\mathcal Q \subset \mathbb{R}^{\abs{\mathcal A} \times \abs{\mathcal X}}$ the set of behaviors $p$ which admit a quantum representation, i.e., the set of behaviors such that there exist a $k$-partite quantum state and local measurements yielding the outcomes $a$ with probability $p(a \given x)$ when performing the measurements $x$.
It is well-known that the set $\mathcal{Q}$ can be approximated from its exterior (from outside the set) by a series of semidefinite programs (SDP) using the NPA hierarchy \cite{NPA08}.

We define a \emph{Bell expression} as a vector $f\in\mathbb{R}^{\abs{\mathcal A}\times\abs{\mathcal X}}$ with components $f(a,x)$.
The Bell expression $f$ defines a linear form on the set of behaviors $p$ through
\begin{equation}
\label{eq:bell}
f[p]= \sum_{a,x} f(a,x)p(a \given x)\,.
\end{equation}
We refer to $f[p]$ as the expectation of $f$ with respect to the behavior $p$.

We consider here a framework in which the information we have about a Bell device is not necessarily given by the full behavior $p$, but possibly only by the expectation of one or more Bell expressions.
In the following, we thus assume that $t$ Bell expressions $f_\alpha$ ($\alpha=1,\dotsc,t$) have been selected.
(The certifiable randomness will depend on this initial choice of Bell expressions; we discuss this issue later.)
We denote by $\bvec{f}=(f_1,\dotsc,f_t)$ these $t$ Bell expressions and by $\bvec{f}[p]=(f_1[p],\dotsc,f_t[p])$ their expectations with respect to the behavior $p$.
As an example, in a bipartite scenario, we might only know the value of the CHSH expression, in which case $t=1$ and there is a single $f$ defined by $f(a,x)=(-1)^{\mathrm{a}_1+\mathrm{a}_2+\mathrm{x}_1\mathrm{x}_2}$.
But the framework is also applicable when $\bvec{f}[p]$ corresponds to the full set $p$ of probabilities.
One simply needs to consider $\abs{\mathcal A}\times \abs{\mathcal X}$ expressions, one for each pairing $(a,x)$, which are defined by $f_{a,x}(a',x')=\delta_{(a,x),(a',x')}$, so that $f_{a,x}[p]=p(a \given x)$.

Of course, in a DI protocol, we are not actually given $\bvec{f}[p]$; we must instead estimate it by performing sequential measurements.
We are thus led to consider a Bell device which is used $n$ times in succession.
We write $\vec{x}=(x_1,\dotsc,x_n)$ and $\vec{a}=(a_1,\dotsc,a_n)$ for the corresponding sequence of inputs and outputs and $\vec{x}_j=(x_1,\dotsc,x_j)$ and $\vec{a}_j=(a_1,\dotsc,a_j)$ for the sequences of inputs and outputs up to, and including, round $j$.

We write $P(\vec{a} \given \vec{x})$ for the conditional probabilities of obtaining the sequence of outputs $\vec{a}$ given a certain sequence of inputs $\vec{x}$.
Note that we use an upper-case $P$ to denote the $n$-round behavior of the boxes and lower-case $p$'s for single-round behaviors.
We assume that the Bell device is probed using inputs $\vec{x}$ distributed according to a probability distribution $\Pi(\vec{x})$.
We will consider, in particular, the case where at each round the inputs are selected  according to identical and independent distributions $\pi(x)$, so that $\Pi(\vec{x})=\prod_{j=1}^n \pi(x_j)$ (though this condition can actually be slightly relaxed in the results that follow).
The full (non-conditional) $n$-round probabilities are thus given by $P(\vec{a},\vec{x})=P(\vec{a} \given \vec{x})\,\Pi(\vec{x})$.
We denote by $P_{AX}$ and $P_{A \given X}$ the distributions corresponding to the probabilities $P(\vec{a},\vec{x})$ and $P(\vec{a} \given \vec{x})$, respectively.

The only assumption we make about the Bell device is that at each round it is characterized by a joint entangled quantum state and a respective set of local measurement operators for each box.
Each set of local measurement operators can depend on the past inputs and outputs of all $k$ boxes (separated boxes can thus freely communicate between measurement rounds), but does not depend on future inputs (inputs are thus selected independently of the state of the device) or inputs of the $k-1$ other boxes in the same round.
Mathematically, this means that we can write $P(\vec{a} \given \vec{x})=\prod_{j=1}^n P(a_j \given x_j,\vec{a}_{j-1},\vec{x}_{j-1})$, and that the (single-round) behavior at round $j$ given the past inputs and outputs $\vec{x}_{j-1}$ and $\vec{a}_{j-1}$, defined as $p_{\vec{a}_{j-1},\vec{x}_{j-1}}(a_j \given x_j)=P(a_j \given x_j,\vec{a}_{j-1},\vec{x}_{j-1})$, should be a valid no-signaling quantum behavior, i.e., $p_{\vec{a}_{j-1},\vec{x}_{j-1}}\in\mathcal{Q}$.

We assume that the internal behavior of the boxes may be classically correlated with a system held by an adversary.
Formally,  these correlations and the adversary's knowledge can be represented through the joint probabilities $P(\vec{a},\vec{x},e)$, where $e$ denotes the adversary's classical side information.
However, in order to keep the notation simple, we do not explicitly include $e$ in the following.
All the reasonings that follow would nevertheless hold, with only minor modifications, if the adversary's classical side information $e$ were explicitly taken into account.
This can be understood by comparing our proofs with those in \cite{PM13}.
Alternatively, $e$ can be formally viewed as an initial input $x_0=e$.

In the following, we sometimes adopt a terminology where the $k$-tuples $x$ and $a$ are referred to as the input and output of (a single-round use of) the Bell device (though of course each consists of the inputs and outputs, respectively, of all $k$ boxes).

\section{A general procedure for DIRNG against classical side-in\-for\-ma\-tion}\label{sec:rand}
In this section, we show how to quantify the randomness produced by $n$ sequential uses of the Bell device based on the Bell expressions $\bvec{f}$.
We follow the approach introduced in \cite{PAM+10,PM13}.
This approach relies on two essential sub-procedures: a first sub-procedure to bound the randomness of single-round behaviors and  a second sub-procedure to estimate a certain quantity involving the Bell expressions $\bvec{f}$.
Given these two ingredients, the single-round randomness bound can, through some simple algebra, be adapted to the $n$-round scenario and related to the actual data obtained in the Bell experiment.

We provide a macro-level description of this approach, which relies only on certain general mathematical properties that these two basic sub-procedures must satisfy, but not on any specifics as to how to implement them.
We will present explicit ways to carry out these sub-procedures in the next two sections.

Intuitively, the output of the Bell device exhibits randomness for some choice of input $\vec x$ if there is no corresponding outcome that is certain to happen, i.e., if $P(\vec{a} \given \vec{x})<1$ for all $\vec{a}\in\mathcal{A}^n$.
Equivalently, we can express this condition by saying that the \emph{surprisals} $-\log_2 P(\vec{a} \given \vec{x})$ are bounded away from zero: $-\log_2 P(\vec{a} \given \vec{x})>0$ for all $\vec{a}\in\mathcal{A}^n$.
Our first aim will thus be to lower-bound these surprisals without making any assumptions regarding the Bell device's behavior apart from the ones stated in Section~\ref{sec:bell}.
We will then see how to turn this bound into a more formal statement in terms of min-entropy.

To bound the $n$-round randomness, we assume the existence of a function $H$ which bounds the \emph{single-round} surprisal $-\log_2 p(a \given x)$ as a function of the Bell expectations $\bvec{f}[p]$.
This is our first ingredient.
We actually require this function to non-trivially bound the surprisals $-\log_2 p(a \given x)$ corresponding to a certain subset $\mathcal{X}_r\subseteq \mathcal{X}$ of all possible inputs.
This is because for certain behaviors, some inputs $x$ lead to less predictable outputs than those resulting from other inputs, and we would therefore prefer to focus on these inputs only.
(As will be elaborated on later, the amount of certifiable randomness generally depends on the choice of $\mathcal{X}_r$.)
Formally, the function $H$, on which our results are based, is defined as follows.
\begin{definition}
\label{def:guess}
Let \(\bvec{f}[\mathcal{Q}] = \{ \bvec{f}[p] : p \in \mathcal{Q} \}\) be the set of Bell expectation vectors compatible with at least one quantum behavior.
A function \(H:\bvec{f}[\mathcal{Q}] \to \ccinterval*{0,\log_2 {\abs{\mathcal A}}}\) is a randomness-bounding (RB) function if it satisfies the following properties:
\begin{enumerate}
\item \(\min_{a\in\mathcal{A},x\in \mathcal{X}_r} \bigl( -\log_2 p(a \given x) \bigr) \geq H(\bvec{f}[p])\) for all \(p\in \mathcal{Q}\).
\item \(H(\bvec{f}[p])\) is a convex function of its argument:
\begin{equation}
 H \bigl(q\, \bvec{f}[p_1]+(1-q)\, \bvec{f}[p_2] \bigr)  \leq q H(\bvec{f}[p_1]) + (1-q) H(\bvec{f}[p_2])
\end{equation}
for any \(0\leq q\leq 1\) and any \(p_1,p_2\in\mathcal{Q}\).
\end{enumerate}\end{definition}

We will also need to compute a lower bound on $H(\bvec{f}[p])$ for all behaviors $p \in \mathcal Q$ such that $\bvec f[p] \in \mathcal V$ for some arbitrary region $\mathcal V \subseteq \mathbb R^t$.
We thus extend our definition of $H$ to sets, such that
\begin{equation}
H(\mathcal{V}) \leq \inf\{H(\bvec{f}[p])\,:\, \bvec{f}[p]\in\bvec{f}[\mathcal{Q}]\cap \mathcal{V}\}\,.
\end{equation}
When $\bvec f[\mathcal Q] \cap \mathcal V = \emptyset$, we define $H(\mathcal V) = 0$.
Furthermore, we define $\eta$ to be a constant such that $\eta \ge \eta^* = \max_{p\in\mathcal{Q}} H(\bvec{f}[p])$.
(In case $\eta^*$ is hard to compute, we may always use $\eta = \log_2\abs{\mathcal A}$.)
We discuss the intuitive interpretation of $H$ and its properties in Section~\ref{sec:sdp}.

Given a RB function $H$, we can now easily lower-bound the $n$-round surprisals:
\begin{lemma}\label{thm:Pax_bound}
Let \(H\) be a RB function.
Then, for any \((\vec{a},\vec{x})\) and any Bell device behavior \(P_{A \given X}\)
\begin{equation}\label{eq:pax_bound_r}
-\log_2 P(\vec{a} \given \vec{x})\geq nH\mleft(\frac{1}{n}\sum_{j=1}^n \bvec{f}\mleft[p_{\vec{a}_{j-1},\vec{x}_{j-1}}\mright]\mright) -\nu(\vec{x})\eta \, ,
\end{equation}
where
\begin{equation}\label{eq:gamma}
\nu(\vec{x})=\sum_{j=1}^n \id_{\mathcal{X}\setminus\mathcal{X}_r}(x_j)
\end{equation}
is the number of \(x_j\) in \(\vec{x}=(x_1,\dotsc,x_n)\) which do not belong to the set \(\mathcal{X}_r\).
\end{lemma}

\begin{proof}[Proof of Lemma \ref{thm:Pax_bound}]
The proof follows essentially the same steps as the proof of Lemma~1 in \cite{PM13}.
The main differences are
(a) that we express the bound eq.~\eqref{eq:pax_bound_r} as a function of $t$ Bell expressions, instead of a single Bell expression, and
(b) that the bound considers explicitly only the randomness from the inputs in $\mathcal{X}_r$.

From our assumptions regarding the Bell device, it follows that for any $(\vec{a},\vec{x})$ we can write
\begin{equation}\label{eq:prodpi}
-\log_2 P(\vec{a} \given \vec{x}) = -\log_2 \prod_{j=1}^n p_{\vec{a}_{j-1},\vec{x}_{j-1}}(a_j \given x_j) =\sum_{j=1}^n -\log_2 p_{\vec{a}_{j-1},\vec{x}_{j-1}}(a_j \given x_j) .
\end{equation}
Each term in the sum such that $x_j\in \mathcal{X}_r$ can be bounded by $H(\bvec{f}[p_{\vec{a}_{j-1},\vec{x}_{j-1}}])$ according to the definition of the function $H$.
If $x_j\notin \mathcal{X}_r$, it is certainly the case that $-\log_2 p_{\vec{a}_{j-1},\vec{x}_{j-1}}(a_j \given x_j) \geq 0 \geq H(\bvec{f}[p_{\vec{a}_{j-1},\vec{x}_{j-1}}]) - \eta$.
We can thus write
\begin{align}
-\log_2 P(\vec{a} \given \vec{x}) & \geq  \sum_{j=1}^n H(\bvec{f}[p_{\vec{a}_{j-1},\vec{x}_{j-1}}]) -\nu(\vec{x}) \eta \\
& \geq nH\mleft(\frac{1}{n}\sum_{j=1}^n \bvec{f}\mleft[p_{\vec{a}_{j-1},\vec{x}_{j-1}}\mright]\mright) -\nu(\vec{x}) \eta , \label{eq:pcxve_bound_r}
\end{align}
where in the last line we have exploited the convexity of $H$.
\end{proof}

Lemma \ref{thm:Pax_bound} tells us how to bound the surprisals $-\log_2 P(\vec{a} \given \vec{x})$ as a function of $\frac{1}{n}\sum_{j=1}^n \bvec{f}[p_{\vec{a}_{j-1},\vec{x}_{j-1}}]$, which can be understood as an $n$-round average Bell expectation, where the average is taken conditioned on past inputs and outputs at each preceding round.
This quantity, however, is not directly observable.
This leads us to introduce the following definition of a confidence region, which is the second ingredient needed in our approach.
\begin{definition}
\label{def:confreg}
A \(1-\epsilon\) confidence region \(\mathcal{V}(\vec{a},\vec{x},\epsilon)\) for \(\frac{1}{n}\sum_{j=1}^n \bvec{f}[p_{\vec{a}_{j-1},\vec{x}_{j-1}}]\) is a subset of \( \mathbb{R}^t \) such that, according to any distribution \(P_{AX}\),
\begin{equation}
\Pr\mleft[\frac{1}{n}\sum_{j=1}^n \bvec{f}\mleft[p_{\vec{a}_{j-1},\vec{x}_{j-1}}\mright]\in \mathcal{V}(\vec a, \vec x, \epsilon)\mright]\geq 1-\epsilon\,.
\end{equation}
We denote by \(V=\{(\vec{a},\vec{x}):\frac{1}{n}\sum_{j=1}^n \bvec{f}[p_{\vec{a}_{j-1},\vec{x}_{j-1}}]\in \mathcal{V}(\vec a, \vec x, \epsilon)\}\) the set of input-output sequences such that \(\frac{1}{n}\sum_{j=1}^n \bvec{f}[p_{\vec{a}_{j-1},\vec{x}_{j-1}}]\) belongs to the confidence region \(\mathcal{V}\).
\end{definition}
(Note that in general $\mathcal{V}$ explicitly depends on $\vec{a}$ and $\vec{x}$, although notation-wise this dependence is sometimes left implicit.)
In other words, for small $\epsilon$ and large $n$, knowing the outcomes $(\vec a, \vec x)$ of $n$ rounds of measurement, one can determine $\mathcal V = \mathcal V(\vec a, \vec x, \epsilon)$ and assert with high confidence that $\frac{1}{n}\sum_{j=1}^n \bvec{f}[p_{\vec{a}_{j-1},\vec{x}_{j-1}}]$ is somewhere in $\mathcal V$, even though its exact value cannot be deduced from $(\vec a, \vec x)$ alone.
The assertion is false if and only if $(\vec a, \vec x) \notin V$, which occurs with a probability smaller than $\epsilon$ by definition.

Combining eq.~\eqref{eq:pcxve_bound_r} with this definition immediately implies the following:
\begin{lemma}\label{thm:pxcve_bound}
Let \(\mathcal{V}\) be a \(1-\epsilon\) confidence region according to Definition \ref{def:confreg}.
Then for any \((\vec{a},\vec{x})\in V\)
\begin{equation}\label{eq:pcxve_bound_r_2}
-\log_2 P(\vec{a} \given \vec{x})\geq nH(\mathcal{V})-\nu(\vec{x})\eta \, .
\end{equation}
\end{lemma}
Lemma~\ref{thm:pxcve_bound} tells us that the surprisal associated to the event $\vec{a}$ given $\vec{x}$ is lower-bounded by a function of $(\vec{a},\vec{x})$, except for a subset of ``bad'' events $\{(\vec{a},\vec{x})\notin V\}$.

One way to deal with these bad events is simply to pretend that the boxes are characterized by a slightly modified behavior $\tilde P$ that yields a new ``abort'' output $\vec{a}=\perp$ when one of the bad events is obtained (while according to $P$, the probability of $\vec{a} = \perp$ is zero).
Effectively, $\tilde P$ can be thought of as post-processed version of the physical behavior $P$.
Though this post-processed version cannot be achieved in practice by the user of the devices (since he does not know the set of bad events), it is well-defined physically (it could for instance be implemented by an adversary having a perfect knowledge of $P$).
The relevant point is that since the probability of these bad events is extremely low for sufficiently small $\epsilon$, the behaviors $P$ and $\tilde P$ are, as shown below, close in variation distance, and analyzing the security using $\tilde P$ instead of $P$ thus yields the same result up to vanishing error terms.
(See \cite{PM13} for a more detailed discussion.)

\begin{lemma} \label{thm:lemma3}
There exists a behavior \(\tilde P_{A \given X}\) such that \(\tilde P_{AX}=\tilde P_{A \given X}\times \Pi_X\) and \(P_{AX}=P_{A \given X}\times \Pi_X\) are \(\epsilon\)-close in variation distance, i.e.,
\begin{equation} \label{eq:variationdistance}
d(\tilde P_{AX},P_{AX}) = \frac{1}{2}\sum_{\vec{a},\vec{x}} \abs{ \tilde P(\vec{a},\vec{x}) - P(\vec{a},\vec{x}) } \leq \epsilon \, ,
\end{equation}
and such that for any \(\vec{a}\neq \perp\)
\begin{equation}\label{eq:Q_bound}
-\log_2 \tilde P(\vec{a} \given \vec{x})\geq nH(\mathcal{V})-\nu(\vec{x})\eta\, .
\end{equation}
\end{lemma}
\begin{proof}[Proof of Lemma \ref{thm:lemma3}]
The proof of this lemma is analogous to that of Lemma~3 in \cite{PM13}.
Define $\tilde P_{A \given X}$ as
\begin{equation}
\tilde P(\vec{a} \given \vec{x}) =
\begin{cases}
P(\vec{a} \given \vec{x}) & \text{if } (\vec{a},\vec{x}) \in V,\\
0 & \text{if } (\vec{a},\vec{x}) \notin V\text{ and }\vec{a} \neq \perp,\\
\sum_{\vec{a} : (\vec{a},\vec{x})\notin V} P(\vec{a} \given \vec{x})& \text{if } \vec{a} =\perp .
\end{cases}
\end{equation}
Eq.~\eqref{eq:variationdistance} follows immediately, and Lemma~\ref{thm:pxcve_bound} implies eq.~\eqref{eq:Q_bound}.
\end{proof}

We can now put a bound on the randomness of the Bell device as follows.
Let $\lambda$ denote the event that $nH(\mathcal{V})-\nu(\vec{x})\eta$ is greater than or equal to some a priori fixed threshold $\tresh$.
Conditioned on $\lambda$ occurring, we can bound the conditional min-entropy of the outputs given the inputs, $H_{\min}(A \given X;\lambda) =  -\log_2 \sum_{\vec{x}} \tilde P(\vec{x} \given \lambda)  \max_{\vec{a}} \tilde P(\vec{a} \given \vec{x};\lambda)$, as follows (see \cite{KRS09} for a more detailed discussion of the concept of min-entropy and its relevance in our context):
\begin{align}
H_{\min}(A \given X;\lambda) &=  -\log_2 \sum_{\vec{x}} \tilde P(\vec{x} \given \lambda)  \max_{\vec{a}} \tilde P(\vec{a} \given \vec{x};\lambda)\\
& = -\log_2 \sum_{\vec{x}} \frac{\tilde P(\vec{x} \given \lambda)}{\tilde P(\lambda \given \vec{x})}  \max_{\vec{a}\in \Lambda_{\vec{x}}} \tilde P(\vec{a} \given \vec{x})\\
& \geq -\log_2 \sum_{\vec{x}} \frac{\tilde P(\vec{x})}{\tilde P(\lambda)} 2^{-\tresh}\\
 & \geq \tresh -\log_2 \frac{1}{\tilde P(\lambda)}\,.\label{eq:h}
\end{align}
In the second line we defined $\Lambda_{\vec{x}}$ as the set of $\vec{a}$'s such that the event $\lambda$ occurs given $\vec{x}$, and in the third line we used eq.~\eqref{eq:Q_bound} and the fact that $nH(\mathcal{V})-\nu(\vec{x})\eta\geq\tresh$ by the definition of $\lambda$.
Comparing $\tilde P(\lambda)$ to some positive $\epsilon'$ directly implies the following result:
\begin{theorem}\label{thm:Theo1}
Let \(\epsilon\) and \(\epsilon'\) be two positive parameters, let \(\tresh\) be some threshold, and let \(\lambda\) be the event that \(nH(\mathcal{V}) -\nu(\vec{x})\eta \geq \tresh\), where \(\mathcal{V}\) is a \(1-\epsilon\) confidence region according to Definition~\ref{def:confreg}.
Then the behavior \(P_{AX}\) is \(\epsilon\)-close to a behavior \(\tilde P_{AX}\) such that, according to \(\tilde P_{AX}\),
\begin{enumerate}
\item either \(\Pr(\lambda) \leq \epsilon'\),
\item or \(H_{\min}(A \given X;\lambda)\geq \tresh-\log_2\frac{1}{\epsilon'}\).
\qed
\end{enumerate}
\end{theorem}
The meaning of this result is as follows.
Suppose that we are able to compute a RB function according to Definition~\ref{def:guess} and, from the results $(\vec a, \vec x)$ of $n$ rounds of measurements, a $1-\epsilon$ confidence region according to Definition~\ref{def:confreg}.
We may thus compute the value of $nH(\mathcal{V})$ and check whether it is above the chosen threshold $\tresh$, i.e., whether the event $\lambda$ occurred.

The given physical device that we used to generate the results $(\vec a, \vec x)$ is characterized by an unknown behavior $P$.
The theorem indirectly characterizes the behavior $P$, by showing the existence of an $\epsilon$-close behavior $\tilde P$, where $\epsilon$ can be chosen arbitrarily small.
The probability difference between the two distributions is thus at most $\epsilon$ for any event, and $P$ and $\tilde P$ are almost indistinguishable.
The theorem states that, assuming that the event $\lambda$ occurs, the behavior $\tilde P$ is one of two possible kinds.

The first possibility if the event $\lambda$ is observed is that the conditional min-entropy of $\tilde P$ is higher than $\tresh - \log_2 \frac{1}{\epsilon'}$.
This implies that $\tilde P$ contains extractable randomness: one can use a randomness extractor to process the raw outputs $\vec{a}$ and obtain a final string of bits, which is close to uniformly random according to $\tilde P$ and whose size is essentially $\tresh-\log_2 \frac{1}{\epsilon'}$ (the length and randomness of the output string will also depend on a security parameter $\epsilon_{\mathrm{ext}}$ of the extractor itself) \cite{DPVR12} .
Since $P$ is $\epsilon$-close to $\tilde P$, it follows that the output string will also be essentially uniformly random according to the actual behavior $P$ of the device (see Section~III.D of \cite{PM13} for details).

The second possibility is that the event $\lambda$ occurred while being very unlikely: according to $\tilde P$, \(\Pr(\lambda) \leq \epsilon'\), and thus, according to $P$, \(\Pr(\lambda) \leq \epsilon'+\epsilon\), where $\epsilon'$ can be chosen arbitrarily small.
In this case there is no guaranteed lower bound on the conditional min-entropy.
We cannot, of course, avoid such a possibility.
For instance, a Bell device that simply outputs predetermined bits, which have been chosen uniformly at random by an adversary, will have zero conditional min-entropy, but may still pass any statistical test we can devise with some positive probability.
Nevertheless, in this case, since $\lambda$ is unlikely, the impact on the security of the protocol of (mistakenly) assuming that the conditional min-entropy bound of the Theorem holds, will be negligible.
We refer to Section~III.D of \cite{PM13} for more details on how Theorem~\ref{thm:Theo1} translates to a secure randomness generation protocol.

Note that more generally, one can use a sequence of thresholds $H_0<H_1<\dotsb<H_\ell$, rather than a single threshold.
Theorem~\ref{thm:Theo1} then becomes a set of individual statements, regarding events $\lambda_i$ where $nH(\mathcal V)-\nu(\vec{x})\eta \in \cointerval{H_i,H_{i+1}}$ for $i = 0, \dotsc, \ell-1$.
This means that the protocol admits intermediate thresholds of success leading to increasingly better min-entropy bounds, rather than being a single-threshold, all-or-nothing protocol.

\section{Estimation}\label{sec:estim}
In this section we explicitly illustrate how to construct a confidence region, according to Definition~\ref{def:confreg}, using a straightforward estimator for the Bell expectations $\bvec{f}[p]$ and applying the Azuma-Hoeffding inequality, as proposed in \cite{PAM+10}.
Note that it is possible to use other (tighter) concentration inequalities than the Azuma-Hoeffding inequality.
In particular, we do not claim our specific choice to be optimal for a finite number of rounds $n$.

Let $(\vec{a},\vec{x})$ be the output-input sequence obtained in a certain realization of the $n$-round protocol.
We define the \emph{observed frequencies} as an estimation of the average behavior of the device based on the observed data,
\begin{equation}
\label{eq:freq}
\hat p(a \given x) = \frac{\#(a,x)}{n\pi(x)} \,,
\end{equation}
where $\#(a,x)$ is the number of occurrences of the output-input pair $(a,x)$ in the $n$ rounds.
As with probabilities, we refer to the full set of observed frequencies $(\hat p(a \given x))$ as a vector $\hat p$.

We define resulting estimators for the Bell expressions by substituting $\hat p$ for $p$ in \eqref{eq:bell}:
\begin{equation}\label{eq:estimator}
\bvec{f}[\hat p]
 = \sum_{a,x} \bvec{f}(a,x) \hat p(a \given x)
 = \frac{1}{n}\sum_{i=1}^n \frac{\bvec{f}(a_i,x_i)}{\pi(x_i)} \,.
\end{equation}
To ease the notation, in the following we sometimes write $\hat{\bvec{f}}$ instead of $\bvec f[\hat p]$.
It should be kept in mind that $\hat p$ and $\hat{\bvec{f}}$ are random variables, being functions of the observed event $(\vec{a},\vec{x})$.

As shown in \cite{PAM+10}, a simple application of the Azuma-Hoeffding inequality yields the following result:
\begin{lemma}\label{thm:azuma0}
For any \(\alpha=1,\dotsc,t\), let \(\epsilon_\alpha^\pm>0\) and let
\begin{equation}\label{eq:muxv}
\mu_\alpha^\pm = \gamma_\alpha \sqrt{\frac{2}{n}\ln \frac{1}{\epsilon_\alpha^\pm}}\,,
\end{equation}
where
\begin{equation}\label{eq:nu}
\gamma_\alpha \geq \max_{p\in\mathcal{Q}}\,\max_{a,x} \abs*{\frac{f_\alpha(a,x)}{\pi(x)}-f_\alpha[p]}\,.
\end{equation}
Then
\begin{equation}
\Pr\mleft[\frac{1}{n}\sum_{j=1}^n f_\alpha[p_{\vec{a}_{j-1},\vec{x}_{j-1}}]\leq \hat f_\alpha+\mu_\alpha^+ \mright]\geq 1-\epsilon_\alpha^+
\end{equation}
and
\begin{equation}
\Pr\mleft[\frac{1}{n}\sum_{j=1}^n f_\alpha[p_{\vec{a}_{j-1},\vec{x}_{j-1}}]\geq \hat f_\alpha-\mu_\alpha^- \mright]\geq 1-\epsilon_\alpha^-\,.
\end{equation}
\end{lemma}
Lemma~\ref{thm:azuma0} simply states that with high probability the $n$-round average $\frac{1}{n}\sum_{j=1}^n f_\alpha[p_{\vec{a}_{j-1},\vec{x}_{j-1}}]$, conditioned on the past, is no greater (no smaller) than the observed value $\hat f_\alpha$ plus (minus) some deviation $\mu_\alpha^+$ ($\mu_\alpha^-$).
This deviation tends to zero as $1/\sqrt{n}$ and directly depends on the quantity  $\gamma_\alpha$,
which represents an upper bound on the maximum possible value of $\abs{f_\alpha(a,x)/\pi(x) - f_\alpha[p]}$, that is to say, the maximal extent to which the random variable $f_\alpha(a,x)/\pi(x)$ can differ from its expectation $f_\alpha[p]$.
In other words, $\gamma_\alpha$ bounds the possible statistical fluctuations which our observations can be subject to.
A specific value for $\gamma_\alpha$ is given by
\begin{align}\label{eq:nu_tilde}
\tilde \gamma_\alpha = \max \mleft\lbrace \max_{a,x} \frac{f_\alpha(a,x)}{\pi(x)} - \min_{p\in\mathcal{Q}} f_\alpha[p],\> \max_{p\in\mathcal{Q}} f_\alpha[p] - \min_{a,x} \frac{f_\alpha(a,x)}{\pi(x)} \mright\rbrace.
\end{align}
The terms $\max_{a,x} \frac{f_\alpha(a,x)}{\pi(x)}$ and $\min_{a,x} \frac{f_\alpha(a,x)}{\pi(x)}$ are easy to calculate,
while the terms $\max_{p\in\mathcal{Q}} f_\alpha[p]$ and $\min_{p\in\mathcal{Q}} f_\alpha[p]$ can be computed through SDP using a NPA relaxation \cite{NPA08}.

We can combine the above upper and lower bounds for all $\alpha$ through a union bound to get the following confidence region:
\begin{lemma}\label{thm:azuma}
Given \(\epsilon_\alpha^\pm \geq 0\) for \(\alpha=1,\dotsc,t\), let \(\hat{\bff}^\pm\) be the vector \((\hat{f}^\pm_1,\dotsc,\hat{f}^\pm_t)\), where
\begin{align}\label{eq:muxvpm}
\hat{f}^\pm_\alpha =
\begin{cases}
\hat f_\alpha \pm \gamma_\alpha \sqrt{\frac{2}{n}\ln \frac{1}{\epsilon^\pm_\alpha}} &\text{if } \epsilon^\pm_\alpha>0 \,,\\
\pm\infty &\text{if } \epsilon^\pm_\alpha=0\,,
\end{cases}
\end{align}
with \(\hat{f}_\alpha\) as defined in eq.~\eqref{eq:estimator} and \(\gamma_\alpha\) as defined in eq.~\eqref{eq:nu}.

Let the confidence region
\begin{equation}\label{eq:confreg}
\mathcal{V}=\freg=\{\bvec{f}\in\mathbb{R}^t\,:\,\hat{\bff}^-\leq \bvec{f}\leq \hat{\bff}^+\}.
\end{equation}
Then
\begin{equation}
\Pr\mleft[\frac{1}{n}\sum_{j=1}^n \bvec{f}[p_{\vec{a}_{j-1},\;\vec{x}_{j-1}}]\in \freg\mright]\geq 1-\epsilon \, ,
\end{equation}
where \(\epsilon=\sum_{\alpha=1}^t(\epsilon^+_\alpha+\epsilon^-_\alpha)\).
\qed
\end{lemma}
In eq.~\eqref{eq:confreg} the inequalities $\hat{\bff}^-\leq \bvec{f}\leq \hat{\bff}^+$---as all other vector inequalities in this paper---should be understood to hold component-wise, i.e., $\hat{f}^-_\alpha \leq f_\alpha \leq \hat{f}^+_\alpha$ for all $\alpha$.

Note that when $\epsilon^+_\alpha=0$ (or  $\epsilon^-_\alpha=0$), we are simply not putting any bound on $\frac{1}{n}\sum_{j=1}^n \bvec{f}[p_{\vec{a}_{j-1},\vec{x}_{j-1}}]$ from above (or below).
Indeed, it is not always useful to bound a Bell expression from both directions.
Consider, for instance, the CHSH expression.
It is well-known that the amount of certifiable randomness increases with the absolute value of the CHSH violation, increasing from $2$ (the maximal local value) to $2\sqrt{2}$ (the maximal quantum value) and from $-2$ (the minimal local value) to $-2\sqrt{2}$ (the minimal quantum value).
If we are estimating the randomness produced by our Bell device based only on the CHSH expression $f_\rchsh$, and strongly expect the CHSH expectation to be in the region $\ccinterval{2,2\sqrt{2}}$, then it is certainly desirable to lower-bound it as accurately as possible.
However, we have no interest in knowing that it is smaller than some value (since the randomness which can be certified is only affected by the lower bound in the region).
For a given $\epsilon=\epsilon^{+}_\rchsh+\epsilon^-_\rchsh$, we are therefore interested in setting $\epsilon^+_\rchsh=0$, so that $\epsilon^-_\rchsh$ is as large as possible, and thus $\hat{f}^-_\rchsh$ is as close as possible to $\hat f_\rchsh$.
However, if we have no a priori reason to expect the CHSH expression to lie in one region or the other, $\epsilon^\pm_\rchsh=\epsilon/2$ is the most natural choice.

\section{Bounding single-round randomness}\label{sec:sdp}
In Section~\ref{sec:rand} we showed how to put a bound on the randomness produced by a Bell device which is used $n$ times in succession, given a RB function $H$.
We now discuss how we can explicitly compute such a function.

The function $H$ is defined through two properties, as specified in Definition~\ref{def:guess}.
The first one is the condition that $-\log_2 p(a \given x)\geq H(\bvec{f}[p])$ for all $a\in\mathcal{A}$, all $x\in \mathcal{X}_r$, and all $p\in \mathcal{Q}$.
The optimal function satisfying this first condition is simply given by
\begin{equation}
\begin{IEEEeqnarraybox}[][c]{lCl'l}
\label{eq:h1}
\tilde H(\bvec{f}[p])  & = &           \min_{a\in\mathcal{A},x\in\mathcal{X}_r} \min_{p'} & -\log_2 p'(a \given x)\\
                  &   & \text{subject to} & \bvec f[p']= \bvec f[p] , \quad p'\in \mathcal{Q}\,.
\end{IEEEeqnarraybox}
\end{equation}
Alternatively, we can pass the $-\log_2$ to the left of the minimizations, which then become maximizations, and we can thus write $\tilde H(\bvec{f}[p])=-\log_2 \tilde G(\bvec{f}[p])$, where
\begin{equation}
\begin{IEEEeqnarraybox}[][c]{lCl'l}
\label{eq:guess1}
\tilde G(\bvec{f}[p])  & = &        \max_{a\in\mathcal{A},x\in\mathcal{X}_r} \max_{p'} & p'(a \given x)\\
                &   & \text{subject to} & \bvec f[p']= \bvec f[p] , \quad p'\in \mathcal{Q}\,.
\end{IEEEeqnarraybox}
\end{equation}

The functions $\tilde H$ and $\tilde G$ defined in this way have an intuitive interpretation.
For a fixed behavior $p$ and a fixed input $x$, $\tilde H=\min_{a\in\mathcal{A}}( -\log_2 p(a \given x) )$ is simply the min-entropy of the distribution $\{p(a \given x)\}_{a\in\mathcal{A}}$, while $\tilde G=2^{-\tilde H}=\max_{a\in\mathcal{A}} p(a \given x)$ is the associated guessing probability, i.e., the optimal probability to correctly guess the output $a$ given that we know that it is drawn from the distribution $\{p(a \given x)\}_{a\in\mathcal{A}}$.
Both these quantities represent measures of the output randomness.
However, we are generally interested in bounding the output randomness not only for a single input $x$, but simultaneously for a subset $\mathcal{X}_r$ of all the inputs.
In addition, we assume in the DI spirit that the full behavior $p$ of our Bell device is generally not known, and that the device is characterized only by the Bell expectations $\bvec{f}[p]$.
Taking the worst case of $\tilde H$ and $\tilde G$ over all inputs $x\in\mathcal{X}_r$ and all quantum behaviors $p$ compatible with the Bell expectations $\bvec{f}[p]$ leads to \eqref{eq:h1} and \eqref{eq:guess1}.

The second requirement in Definition~\ref{def:guess} is that $H$ should be a convex function.
This property is used in Lemma~\ref{thm:Pax_bound} to bound the randomness produced from $n$ successive measurement rounds.
However, the function defined by \eqref{eq:h1} is not necessarily convex.
For fixed values of $a\in \mathcal{A}$ and $x \in \mathcal{X}_r$, let us denote $\tilde H_{a,x}(\bvec{f}[p])$ the function defined by the interior minimization, i.e., the minimum over $p' \in \mathcal{Q}$ of $-\log_2 p'(a \given x)$ subject to the constraint that $\bvec f[p'] = \bvec f[p]$.
This is a convex minimization program and thus the functions $\tilde H_{a,x}(\bvec{f}[p])$ are all convex.
However, $\tilde H$ is obtained by taking the point-wise minimum $\tilde H(\bvec{f}[{p}])=\min_{a,x} \tilde H_{a,x}(\bvec f[p])$ of these functions, which will generally not be convex (see \cite{BSS14} for a specific example where this happens).
Similarly, the individual functions $\tilde G_{a,x}$ defined by the interior maximization in \eqref{eq:guess1} are concave, but $\tilde G$ will generally not be.

In order to obtain a convex function, we could simply define a function $H^*$ as the minimum over arbitrary convex combinations of the functions $\tilde H_{a,x}$, i.e., as the convex hull of \eqref{eq:h1}:
\begin{equation}
\begin{IEEEeqnarraybox}[][c]{lCl'l}
\label{eq:h2}
H^*(\bvec{f}[p])  & = &           \min_{\{q_{a,x} , p_{a,x} \}_{a\in\mathcal{A},x\in\mathcal{X}_r}} & \sum_{a\in\mathcal{A},x\in\mathcal{X}_r} q_{a,x} \tilde H_{a,x}(\bvec{f}[{p}_{a,x}])\\
                  &   & \text{subject to} & q_{a,x}\geq 0\,,\\
                  &   & & \sum_{a\in\mathcal{A},x\in\mathcal{X}_r}	 q_{a,x}=1\,,\\
                  &   & & \sum_{a\in\mathcal{A},x\in\mathcal{X}_r} q_{a,x}\bvec f[p_{a,x}]=\bvec f[p]\,.
\end{IEEEeqnarraybox}
\end{equation}
Similarly, the concave hull of \eqref{eq:guess1} is
\begin{equation}
\begin{IEEEeqnarraybox}[][c]{lCl'l}
\label{eq:guess1b}
G(\bvec{f}[p])  & = &           \max_{\{q_{a,x} , p_{a,x}\}_{a\in\mathcal{A},x\in\mathcal{X}_r}} & \sum_{a\in\mathcal{A},x\in\mathcal{X}_r} q_{a,x} \tilde G_{a,x}(\bvec{f}[{p}_{a,x}])\\
                &   & \text{subject to} & q_{a,x}\geq 0\,,\\
                &   & & \sum_{a\in\mathcal{A},x\in\mathcal{X}_r}	 q_{a,x}=1\,,\\
                &   & & \sum_{a\in\mathcal{A},x\in\mathcal{X}_r} q_{a,x}\bvec f[p_{a,x}]=\bvec f[p]\,.
\end{IEEEeqnarraybox}
\end{equation}
Note that it is not true any more that $H^*=-\log_2 G$, but it is easy to see that $H^* \geq H = -\log_2 G$.

Though the function $H^*$ defined through \eqref{eq:h2} is the tightest function satisfying the constraint of Definition ~\ref{def:guess}, it is not easy to deal with numerically because of the presence of the logarithms in the definitions of $\tilde H_{a,x}$.
We will thus instead use the lower-bound $H=-\log_2 G$, which obviously satisfies the first condition of Definition 1 (since $H^*\geq H$) as well as the second one (since $G$ is concave and nonnegative, $H=-\log_2 G$ is convex).
The interest is that the optimization problem \eqref{eq:guess1b} is simpler to evaluate than \eqref{eq:h2}.
Note first that \eqref{eq:guess1b} can be re-expressed as follows by absorbing the weights $q_{a,x}$ in the unnormalized quantum behaviors $\tilde p_{a,x}=q_{a,x}p_{a,x}$:
\begin{align}
\begin{split}\label{eq:guess4}
G(\bvec{f}[p])  = \max_{\{\tilde p_{a,x}\}_{a\in\mathcal{A},x\in\mathcal{X}_r}} &\quad \sum_{a\in\mathcal{A},x\in\mathcal{X}_r}\tilde p_{a,x}(a \given x)\\
\text{subject to} & \quad \sum_{a\in\mathcal{A},x\in\mathcal{X}_r} \tr[\tilde p_{a,x}]  =1,\\
 &\quad  \sum_{a\in\mathcal{A},x\in\mathcal{X}_r}\, \bvec{f}[\tilde p_{a,x}]=\bvec{f}[p] ,\\
& \quad \tilde p_{a,x}\in \tilde{\mathcal{Q}} \quad \forall a\in\mathcal{A}, \forall x\in \mathcal{X}_r \,.
\end{split}
\end{align}
In the above formulation, $\tilde{\mathcal{Q}}$ denotes the set of unnormalized quantum behaviors, the conditions $q_{a,x}\geq 0$ and $p_{a,x} \in\mathcal{Q}$ are equivalent to the single condition $\tilde p_{a,x} \in \tilde{\mathcal{Q}}$, and the condition $\sum_{a,x} q_{a,x}=1$ becomes $\sum_{a,x} \tr[\tilde p_{a,x}]=1$ where $\tr[p]=\sum_{a}p(a \given x)$ is the norm of $p$ (the expression $\tr[p]$ is independent of the choice of $x$, and it is equal to $1$ for normalized behaviors).
Problem \eqref{eq:guess4} cannot be solved in general since the set $\tilde{\mathcal{Q}}$ is hard to characterize, but it can be replaced with one of its NPA relaxations, in which case it becomes a SDP (since apart from the condition $\tilde p_{a,x}\in\tilde{\mathcal{Q}}$ all constraints and the objective function are linear).
This will in general only yield an upper bound on the optimal value $G$ (and thus a lower bound on $H^*$), but this is entirely sufficient for our purpose.

In the case where the set $\mathcal{X}_r$ contains a single input $x$, the optimization problem \eqref{eq:guess4} is essentially identical to the one introduced in \cite{NPS14,BSS14} and corresponds to maximizing an adversary's average guessing probability over all possible quantum strategies (the difference with \cite{NPS14,BSS14} is that we characterize the devices through an arbitrary number of Bell expectations $\bvec{f}[p]$, rather than a single Bell expression or the full set of probabilities $p(a \given x)$).
The general form \eqref{eq:guess4}, however, also applies to the case where $\mathcal{X}_r$ contains more than one input and represents one possible way to characterize the randomness of a subset of inputs (other suggestions have been made in \cite{BSS14}; the main reason for the present choice is that it satisfies the mathematical properties that are needed in our $n$-round analysis).
In the following, we refer to the function $G$ given by \eqref{eq:guess4} as the guessing probability of the behavior characterized by $\bvec{f}[p]$.

To apply our $n$-round analysis, we actually do not need to compute the value $H(\bvec{f}[p])=-\log_2 G(\bvec{f}[p])$ for a fixed value $\bvec{f}[p]$, but instead its worst-case bound over all quantum behaviors $p \in \mathcal Q$ for which $\bvec f[p] \in \mathcal{V}$.
If the confidence region $\mathcal{V}$ is defined as an interval $\freg$, as in the preceding section, this can simply be cast as the following optimization problem:
\begin{align}
\begin{split}\label{eq:guess5}
G(\freg)  = \max_{\{\tilde p_{a,x}\}_{a\in\mathcal{A},x\in\mathcal{X}_r}} &\quad \sum_{a\in\mathcal{A},x\in\mathcal{X}_r}\tilde p_{a,x}(a \given x)\\
\text{subject to} & \quad \sum_{a\in\mathcal{A},x\in\mathcal{X}_r} \tr[\tilde p_{a,x}]  =1,\\
 &\quad  \hat{\bff}^-\leq\sum_{a\in\mathcal{A},x\in\mathcal{X}_r}\, \bvec{f}[\tilde p_{a,x}]\leq \hat{\bff}^+ ,\\
& \quad \tilde p_{a,x}\in \tilde{\mathcal{Q}} \quad \forall a\in\mathcal{A}, \forall x\in \mathcal{X}_r .
\end{split}
\end{align}
Again this problem admits a SDP relaxation through the NPA hierarchy.
When $\freg \cap \bvec f[\mathcal Q] = \emptyset$, the optimization problem is infeasible.
In accordance with Definition~\ref{def:guess}, in that case we let $G(\freg) = 1$ or $H(\freg)=0$.

We conclude this discussion by noting that in specific cases such as that of \cite{PAM+10}, where $\bvec{f}$ is a single CHSH expression, the symmetries under relabelings of inputs and outputs imply that the formulations \eqref{eq:h1}, \eqref{eq:guess1}, \eqref{eq:h2}, and \eqref{eq:guess1b} are equivalent, since \eqref{eq:h1} is already convex.
In such cases, our RB function is the tightest function that satisfies Definition~\ref{def:guess}, by virtue of \eqref{eq:h1} being the tightest function that satisfies condition~1 of the Definition.

In the Appendix, we provide more intuition about the above problems by considering their dual formulations.
We also discuss in more detail their link with \cite{NPS14,BSS14}.

\section{Summary of the protocol}\label{sec:prot}
In the two preceding sections, we have specified a way of bounding the randomness within a $(1-\epsilon)$ confidence region $\mathcal{V} = \freg$ around the observed statistic $\bvec f[\hat p]$.
We can thus apply \autoref{thm:Theo1} to bound the min-entropy of the output string obtained after $n$ uses of the device.
Processing this raw string with a suitable extractor finally leads to a uniformly random and private string.
The resulting protocol is summarized in \autoref{fig:roverlineprot}.

\begin{figure}[htp]
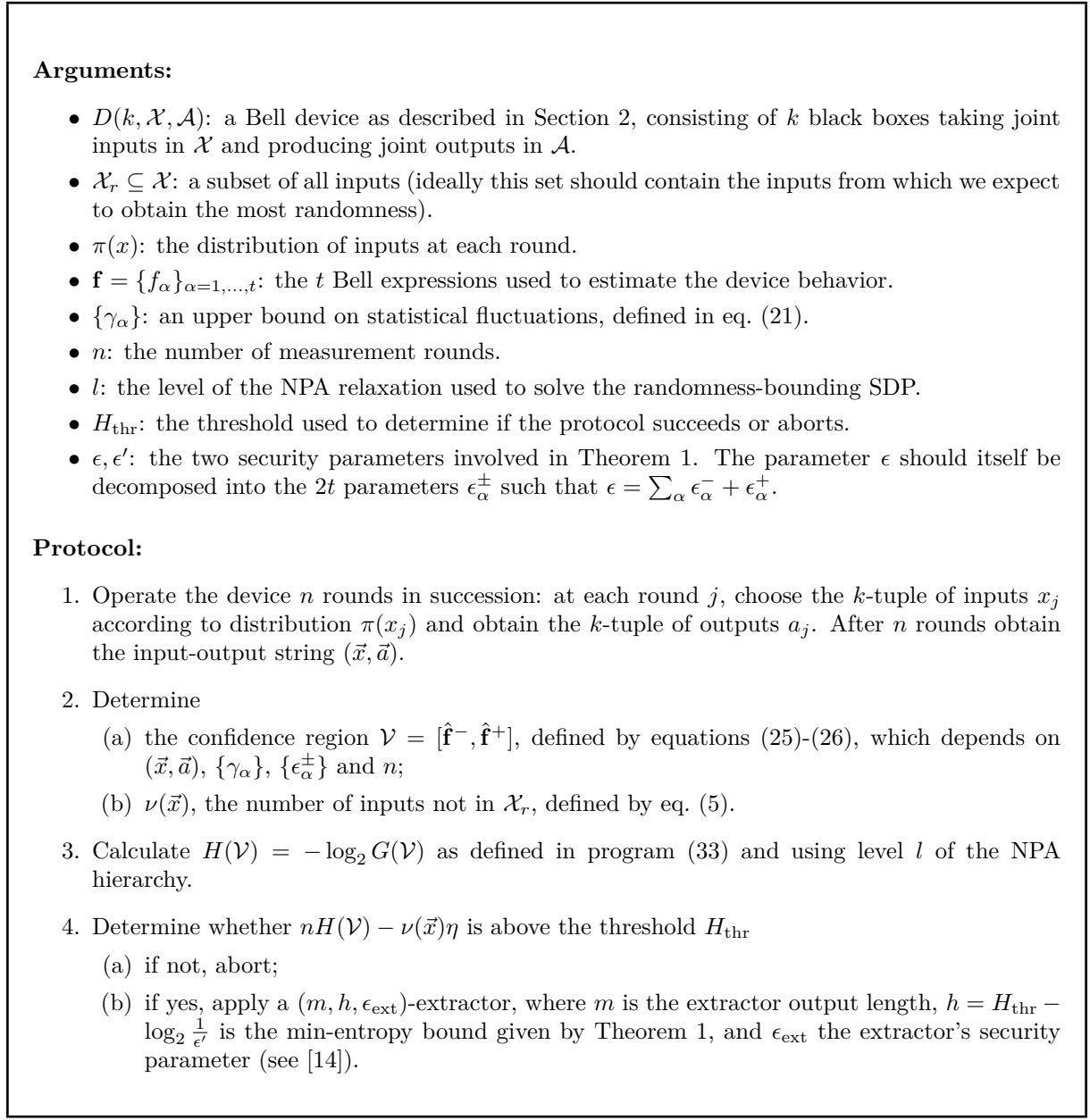

\begin{center}
\setlength{\fboxsep}{10pt}
\setlength{\fboxrule}{1pt}
\fbox{\parbox{\dimexpr\linewidth-2\fboxsep-2\fboxrule}{

\paragraph{Arguments:}
\begin{itemize}[itemsep=-1pt]%
\item $D(k,\mathcal{X},\mathcal{A})$: a Bell device as described in Section~\ref{sec:bell}, consisting of $k$ black boxes taking joint inputs in $\mathcal{X}$ and producing joint outputs in $\mathcal{A}$.

\item $ \mathcal{X}_r \subseteq \mathcal{X}$: a subset of all inputs (ideally this set should contain the inputs from which we expect to obtain the most randomness).

\item $\pi(x)$: the distribution of inputs at each round.

\item $ \bvec f = \{f_\alpha\}_{\alpha=1,\dotsc,t}$: the $t$ Bell expressions used to estimate the device behavior.

\item $\{\gamma_\alpha\}$: an upper bound on statistical fluctuations, defined in eq.~\eqref{eq:nu}.

\item $n$: the number of measurement rounds.

\item $l$: the level of the NPA relaxation used to solve the randomness-bounding SDP.

\item $\tresh$: the threshold used to determine if the protocol succeeds or aborts.

\item $\epsilon, \epsilon'$: the two security parameters involved in Theorem~1.
The parameter $\epsilon$ should itself be decomposed into the $2t$ parameters $\epsilon_\alpha^\pm$ such that $\epsilon = \sum_\alpha \epsilon^-_\alpha + \epsilon^+_\alpha$.

	\end{itemize}

\paragraph{Protocol:}

\begin{enumerate}%

\item Operate the device $n$ rounds in succession: at each round $j$, choose the $k$-tuple of inputs $x_j$ according to distribution $\pi(x_j)$ and obtain the $k$-tuple of outputs $a_j$.
After $n$ rounds obtain the input-output string $(\vec{x},\vec{a})$.

\item Determine
\begin{enumerate}[topsep=0pt]
\item the confidence region $\mathcal{V}=\freg$, defined by equations \eqref{eq:muxvpm}-\eqref{eq:confreg}, which depends on $(\vec{x},\vec{a})$, $\{\gamma_\alpha\}$, $\{\epsilon_\alpha^\pm\}$ and $n$;
\item $\nu(\vec{x})$, the number of inputs not in $\mathcal{X}_r$, defined by eq.~\eqref{eq:gamma}.
\end{enumerate}

\item Calculate $H(\mathcal{V})=-\log_2 G(\mathcal{V})$ as defined in program \eqref{eq:guess5} and using level $l$ of the NPA hierarchy.

\item Determine whether $nH(\mathcal{V})-\nu(\vec{x})\eta $ is above the threshold $\tresh$
	\begin{enumerate}[topsep=0pt]
	\item if not, abort;
	\item if yes, apply a $(m,h,\epsilon_{\mathrm{ext}})$-extractor, where $m$ is the extractor output length, $h = \tresh - \log_2 \frac{1}{\epsilon'}$ is the min-entropy bound given by Theorem~\ref{thm:Theo1}, and $\epsilon_{\mathrm{ext}}$ the extractor's security parameter (see \cite{PM13}).
	\end{enumerate}
\end{enumerate}
}}

\caption{DIRNG protocol following from Theorem~\ref{thm:Theo1} (Section~\ref{sec:rand}) and from the results of Sections~\ref{sec:estim}, \ref{sec:sdp}.}
\label{fig:roverlineprot}
\end{center}
\end{figure}

As we noted in Section~\ref{sec:rand}, one can define a similar protocol based on a sequence of thresholds $H_0<H_1<\dotsb<H_\ell$ rather than a single one, introducing intermediate levels of success in the protocol.
One advantage of this is that we do not need to determine what threshold we expect the device to reach and risk failing the protocol with high probability if we overestimated $\tresh$.
See Section~III.D of \cite{PM13} for details.

\section{Discussion}
\label{sec:discussion}
We have introduced a family of protocols, each characterized by a choice of $t$ Bell expressions $f_\alpha$, a randomness-generating input set $\mathcal{X}_r$, and an input distribution $\pi(x)$.
This family contains as a special case the protocols introduced in \cite{PAM+10, PM13, FGS13}, which correspond to the case where a single Bell expression $f$ is used ($t=1$) and where the randomness-bounding function covers all inputs ($\mathcal{X}_r=\mathcal{X}$).
The main novelty introduced in the present work is that we can take into account information from more Bell expressions $(t\geq 1)$ and can tailor the randomness analysis to a subset of all possible inputs ($\mathcal{X}_r\subseteq\mathcal{X}$).

In order to discuss these new aspects, in the following sections we illustrate our protocol on a concrete example.
The scenario in this example has two parties ($k=2$), two measurement settings per party ($\mathcal X = \{0,1\}^2$), and two outcome possibilities per measurement ($\mathcal A = \{0,1\}^2$).
We consider a device behavior
\begin{equation}
\label{eq:behavior}
p = v p_\text{ext} + (1-v) u
\end{equation}
for $v = 0.99$,
arising from a mixture of white noise $u$ and the extremal quantum behavior $p_\text{ext}$ that achieves maximal violation of the $I_1^\beta$ tilted-CHSH inequality introduced in \cite{AMP12}, with $\beta = 2\cos(2\theta)/\sqrt{1+\sin^2(2\theta)}$ for $\theta = \pi/8$.
The tilted-CHSH expression is defined as
\begin{equation}
\label{eq:tilted-chsh}
I_1^\beta = \beta \mean{A_0} + \mean{A_0 B_0} + \mean{A_0 B_1} + \mean{A_1 B_0} - \mean{A_1 B_1} \,.
\end{equation}
The extremal behavior can be achieved by a pair of partially entangled qubits $\ket\phi = \cos\theta \ket{00} + \sin\theta \ket{11}$ measured with observables
\begin{IEEEeqnarray}{rL"rL}
\IEEEyesnumber
\label{eq:operators}
\IEEEyessubnumber*
	A_0 &= \sigma_x \,, &
	A_1 &= \sigma_z \,, \\
	B_0 &= \cos\mu \, \sigma_z + \sin\mu \, \sigma_x \,, &
	B_1 &= \cos\mu \, \sigma_z - \sin\mu \, \sigma_x \,,
\end{IEEEeqnarray}
with $\tan\mu = \sin 2\theta$.
(Note the difference in notation from \cite{AMP12}: we relabelled the inputs $1$ and $2$ to $0$ and $1$, respectively.)

The resulting correlations have the property of giving more predictable outcomes for a subset of measurement inputs.
For $\theta = \pi/8$, the two measurement settings that give more predictable outcomes, $x = (0,0)$ and $x = (0,1)$, have a guessing probability of about $0.775$ in the ideal ($v=1$) case where $I_1^\beta$ is maximally violated.
On the other hand, the two measurement settings with less predictable outcomes, $x = (1,0)$ and $x = (1,1)$, have guessing probabilities of about $0.496$.

In the analysis of randomness, we will consider two choices for the randomness generating subset $\mathcal X_r$: the full input set $\mathcal X_r = \mathcal X$ and the more restricted choice $\mathcal X_r = \{(1,0)\}$, which is one of the two settings that give less predictable measurements in $p_\text{ext}$.

Furthermore, we will estimate three different Bell expressions, all defined in terms of the following correlators:
\begin{subequations}
\label{eq:correlators}
\begin{align}
\mean{A_\mathrm{x_1}} &= \sum_{\mathrm{a}_1,\mathrm{a}_2,\mathrm{x}_2 \in \{0,1\}} (-1)^\mathrm{a_1} \pi_2(\mathrm x_2 \given \mathrm x_1) p(\mathrm{a}_1\mathrm{a}_2 \given \mathrm{x}_1\mathrm{x}_2) & \mathrm{x}_1 \in \{0,1\} ,\\
\mean{B_\mathrm{x_2}} &= \sum_{\mathrm{a}_1,\mathrm{a}_2,\mathrm{x}_1 \in \{0,1\}} (-1)^\mathrm{a_2} \pi_1(\mathrm x_1 \given \mathrm x_2) p(\mathrm{a}_1\mathrm{a}_2 \given \mathrm{x}_1\mathrm{x}_2) & \mathrm{x}_2 \in \{0,1\} ,\\
\mean{A_\mathrm{x_1}B_\mathrm{x_2}} &= \sum_{\mathrm{a}_1,\mathrm{a}_2 \in \{0,1\}} (-1)^{\mathrm{a}_1+\mathrm{a}_2} p(\mathrm{a}_1\mathrm{a}_2 \given \mathrm{x}_1\mathrm{x}_2) & \mathrm{x}_1,\mathrm{x}_2 \in\{0,1\} . \label{eq:corr2}
\end{align}
\end{subequations}
The weights $\pi_1(\mathrm x_1 \given \mathrm x_2)$ and $\pi_2(\mathrm x_2 \given \mathrm x_1)$ represent the two conditional local input distributions defined with respect to the joint input distribution $\pi(\mathrm x_1 \mathrm x_2)$.\footnote{%
Note that for a no-signaling behavior, we could equivalently use any arbitrary set of probability weights in place of $\pi_1(\mathrm x_1 \given \mathrm x_2)$ or $\pi_2(\mathrm x_2 \given \mathrm x_1)$.
We choose this specific set of weights because it provides a better estimator for the marginal correlators $\mean{A_{\mathrm x_1}}$ and $\mean{B_{\mathrm x_2}}$ when applied to the observed frequencies $\hat p(\mathrm a_1 \mathrm a_2 \given \mathrm x_1 \mathrm x_2)$.
Indeed, it can be seen from the definition of a Bell estimator $f[\hat p]$ in eq.~\eqref{eq:estimator} that the marginal correlators reduce to a natural definition based on locally available data resulting only from the respective party's interaction with their part of the device, namely,
$\mean{A_{\mathrm x_1}} = \sum_{\mathrm a_1} (-1)^{\mathrm a_1} \#(\mathrm a_1,\mathrm x_1)/(n \pi_1(\mathrm x_1))$,
and similarly for $\mean{B_{\mathrm x_2}}$.
}

The expressions we will evaluate are the CHSH expression
\begin{equation}
\label{eq:chsh}
I_\rchsh = \mean{A_0B_0} + \mean{A_0B_1} + \mean{A_1B_0} - \mean{A_1B_1} \,,
\end{equation}
the tilted-CHSH expression $I_1^\beta$ \eqref{eq:tilted-chsh}
and the ``optimal'' expressions for the chosen device behavior \eqref{eq:behavior}:
\begin{equation}
\begin{split}
I_p ={} &
   10.610
  - 1.859 \, \mean{A_0}
  - 1.733 \, \mean{A_1}
  + 0.499 \, \mean{B_0}
  - 2.196 \, \mean{B_1} \\&
  - 3.109 \, \mean{A_0B_0}
  - 2.945 \, \mean{A_0B_1}
  - 2.610 \, \mean{A_1B_0}
  + 4.343 \, \mean{A_1B_1}
\end{split}
\end{equation}
and
\begin{equation}
\begin{split}
I_p^{\text{all}} ={} &
    3.131
  + 0.126 \, \mean{A_0}
  - 0.428 \, (\mean{B_0} + \mean{B_1}) \\&
  - 0.673 \, (\mean{A_0B_0} + \mean{A_0B_1})
  - 1.002 \, (\mean{A_1B_0} - \mean{A_1B_1})
  \,.
\end{split}
\end{equation}
These last two Bell expressions are ``optimal'' Bell expressions in the following sense.
As already observed in \cite{NPS14,BSS14}, the dual of problem \eqref{eq:guess4} (see eq.~\eqref{eq:g_v(p)} in the Appendix), when applied to a device characterized by its full behavior (i.e., $f_{a,x}[p]=p(a \given x)$ so that $\bvec{f}[p] = p$), finds a Bell expression $I_p$ such that the amount of randomness certified from $I_p[p]$ with respect to the measurement setting $x = (1,0)$ is equal to the amount of randomness that can be certified from the entire table of probabilities $p(a \given x)$ (again, with respect to the measurement $x = (1,0)$).
Thus, to each device behavior $p$ is associated a single Bell expression $I_p$ that is optimal for $p$ from the point of view of randomness.%
\footnote{%
More accurately, there exist infinitely many Bell expressions that are equivalent to $I_p$ up to terms that vanish for no-signaling behaviors.
In order to pick one that tolerates the small signaling fluctuations present in our behavior estimator $\hat p$, we run the computation of $I_p$ in the $8$-dimensional space of correlators, rather than the overspecified $16$-dimensional parametrization of quantum behaviors in terms of the probabilities $p(\mathrm a_1 \mathrm a_2 \given \mathrm x_1 \mathrm x_2)$.
We translate this expression back to a unique standard form \eqref{eq:bell} using definition \eqref{eq:correlators} for the correlators.
This ensures that the solution to the dual program \eqref{eq:g_v(p)} picked by our solver among many equivalent expressions does not contain terms that blow up under small signaling fluctuations.
See also \cite{RRMG17} for a finer analysis of noise tolerance in equivalent Bell expressions.
}
Likewise, $I_p^\text{all}$ is defined with respect to all inputs $x \in \mathcal X$ rather than the subset $\{(1,0)\}$.

\subsection{Bounding randomness for all inputs with one Bell expression (\texorpdfstring{$\mathcal X_r = \mathcal X$, $t = 1$}{X\textrinferior{} = X, t = 1})} \label{sec:pra_case}
Before discussing the novelties introduced in this work, let us start by briefly reviewing the case $t=1$ and $\mathcal{X}_r=\mathcal{X}$, which corresponds to the protocols introduced in \cite{PAM+10, PM13, FGS13}.
In this case, $\nu(\vec{x})$, the number of inputs not in $\mathcal{X}_r$, is always equal to zero, and according to Theorem~\ref{thm:Theo1}, the min-entropy of the output string is roughly equal to $nH(\mathcal{V})$.
Furthermore, the confidence region $\mathcal{V}$ reduces to a confidence interval $[\hat f^-,\hat f^+]$ around the estimated Bell violation $\hat f$.
Usually, the values of $\hat f$ that we expect to obtain in the protocol will fall in a region where $H(\hat f)$ is either monotonically increasing  or decreasing with $\hat f$, i.e., the interval is within either the upward- or downward-sloped region of the convex function $H(\hat f)$.
For instance, if $f$ is the CHSH expression, we may assume that the devices have been designed so that with very high probability $\hat f\geq 2$.
In that region, $H(\hat f)$ is indeed increasing with $\hat f$ (i.e., the randomness increases for increasing values of the CHSH expression).
Let us assume for definiteness that $H$ is increasing (the same kind of reasoning can be done if $H$ is decreasing).
Since we are looking for the minimal value of $H$ in the region $\mathcal{V}$ (see Lemma~\ref{thm:pxcve_bound}), it is then sufficient, as done in \cite{PAM+10, PM13, FGS13}, to take a one-sided interval $\cointerval{\hat f^-,\infty}$, and the minimal value of $H$ in the interval will then be $H(\hat f^-)$.
Considering again our CHSH example, we are interested in a guarantee that the CHSH value is above some threshold, which determines the randomness we can certify in the worst case, but it is useless to know that it is bounded from above (see also discussion at the end of Section~\ref{sec:estim}).
Taking the definition eq.~\eqref{eq:muxvpm} for $\hat f^-$, we thus get that the min-entropy of the output string is bounded (roughly speaking%
\footnote{%
Equation~\eqref{eq:Gf-} and the similar approximate bounds that follow should be understood as informal statements giving an order of magnitude for the min-entropy lower bound.
Contrary to the statement of Theorem~\ref{thm:Theo1}, this informal bound directly involves the estimator $\hat f$, which is a random variable.
As such, it might be subject to improbable but extreme fluctuations, in which case the bound does not correctly characterize the device.
In comparison, the min-entropy bound of Theorem~\ref{thm:Theo1} is expressed in terms of a fixed threshold.
Furthermore, the Theorem also accounts for the unlikely event that a device reaches this threshold only by chance.
}%
, and up to the $-\log_2(1/\epsilon')$ correction) as
\begin{align}\label{eq:Gf-}
H_\text{min} \gtrsim nH\mleft( \hat{f} - \gamma \sqrt{\frac{2}{n}\ln \frac{1}{\epsilon}} \mright)\,.
\end{align}
This is precisely the result of \cite{PAM+10, PM13, FGS13}, whose interpretation is quite intuitive:
the min-entropy after $n$ runs is equal to $n$ times the min-entropy for a single run, evaluated on the observed Bell violation $\hat{f}$ offset by a statistical parameter $\mu = \gamma \sqrt{(2/n) \ln (1/\epsilon)}$.
This correction accounts for the fact that even if a device has been built such that it produces a target Bell violation, statistical fluctuations may push the observed violation above what is expected.

This statistical correction depends on the security parameter $\epsilon$ and decreases with the number of runs $n$.
It also depends on the prefactor $\gamma$ defined in eq.~\eqref{eq:nu}.
This prefactor depends on the choice of Bell expression $f$, and also importantly on the input distribution $\pi(x)$.

As discussed in \cite{PAM+10,PM13}, the input distribution can be suitably chosen to optimize the ratio $R_{\text{out}}/R_{\text{in}}$ of the randomness that is produced to the randomness that is consumed when choosing the inputs.
The idea is that if at each run one selects with very high probability a given input $x=x^*$, then the resulting distribution $\pi(x)$ can be sampled from a small number of initial uniform bits $R_{\text{in}}$, which should improve the ratio $R_{\text{out}}/R_{\text{in}}$.
However, this will also lower $R_{\text{out}}$ because observations involving the other inputs $x\neq x^*$ will be less frequent, which will reduce the statistical accuracy.
Consider for instance, as in \cite{PAM+10,PM13}, the case where the input $x^*$ is chosen with probability $\pi(x^*)=1-\kappa n^{-\delta}$ for some constants $\kappa$ and $\delta$, and the other inputs are chosen with probability $\pi(x)=\kappa' n^{-\delta}$, where $\kappa'=\kappa/(\abs{\mathcal X}-1)$ for normalization.
Then the initial randomness $R_{\text{in}}$ required to choose the inputs according to this distribution will be of size $O(n^{1-\delta}\ln n^\delta)$ (i.e., roughly $n$ times the Shannon entropy of the input distribution, see Theorem~2 in \cite{PM13}).
On the other hand, according to eq.~\eqref{eq:Gf-}, the output randomness will be of size $\Omega(n)$ as long as the statistical correction, of order $\gamma/\sqrt{n}$, remains bounded by a constant.
Since, according to eq.~\eqref{eq:nu}, $\gamma\simeq 1/(\min_x \pi(x))$ we get that the statistical correction is of order $\gamma/\sqrt{n}=O(n^{\delta-\frac{1}{2}})$ and thus that we should take $\delta\leq \frac{1}{2}$.
We can thus hope at best a quadratic expansion wherein $O(n^{\frac{1}{2}}\ln n^{\frac{1}{2}})$ initial bits are consumed and $\Omega(n)$ are produced.

Note that the initial randomness for choosing the inputs only needs to be random with respect to the devices, but can be publicly announced to the adversary without compromising the privacy of the output string \cite{PM13,CSW14}.
One can thus view the above protocols as producing private randomness from public randomness.
From this perspective, the ``expansion'' efficiency of the protocol is less relevant since the final and initial randomness correspond to different resources that do not necessarily have to be compared on the same footing.

We generated random samples of $n$ input-output pairs $\vec{a},\vec{x}$ from the behavior $p$ corresponding to equation~\eqref{eq:behavior} with the following input distribution
\begin{align}
\label{eq:input-dist}
\pi(x) =
\begin{cases}
1 - \frac{3}{2} n^{-1/5} \quad &\text{if } x =  (1,0) \,, \\
\frac{1}{2} n^{-1/5} \quad &\text{otherwise.}
\end{cases}
\end{align}
Note that as $n$ grows, the input distribution becomes strongly biased to select $x = (1,0)$ most of the time.

We performed this sampling independently for different values of $n$ between $100$ and $3 \times 10^{18}$.
For each value of $n$, we repeated this sampling $300$ times in order to show the variation of our result over several simulations.

The corresponding min-entropy rate bound (that is, \eqref{eq:Gf-} divided by $n$) for $\epsilon=10^{-6}$ is represented in Figure~\ref{fig:1} as a function of the number of runs $n$ for different Bell expressions.
The curves in this plot and the ones that follow (Figures~\ref{fig:1}--\ref{fig:5}) show the values for the first simulation out of the $300$, and the range of values taken over all $300$ simulations is drawn as a shaded area behind each curve.
In some instances, usually for high values of $n$, the area is invisible, which indicates a negligible variation across simulation runs.
All curves are obtained by solving the program~\eqref{eq:guess5} in its dual form~\eqref{eq:noisy_conic_dual3} (see Appendix) at level 2 of the NPA hierarchy.
All optimizations were performed using the \textsc{Matlab} toolboxes \textsc{Yalmip} \cite{yalmip} and SeDuMi \cite{sedumi}.

\begin{figure}[t]
\centering
\includegraphics{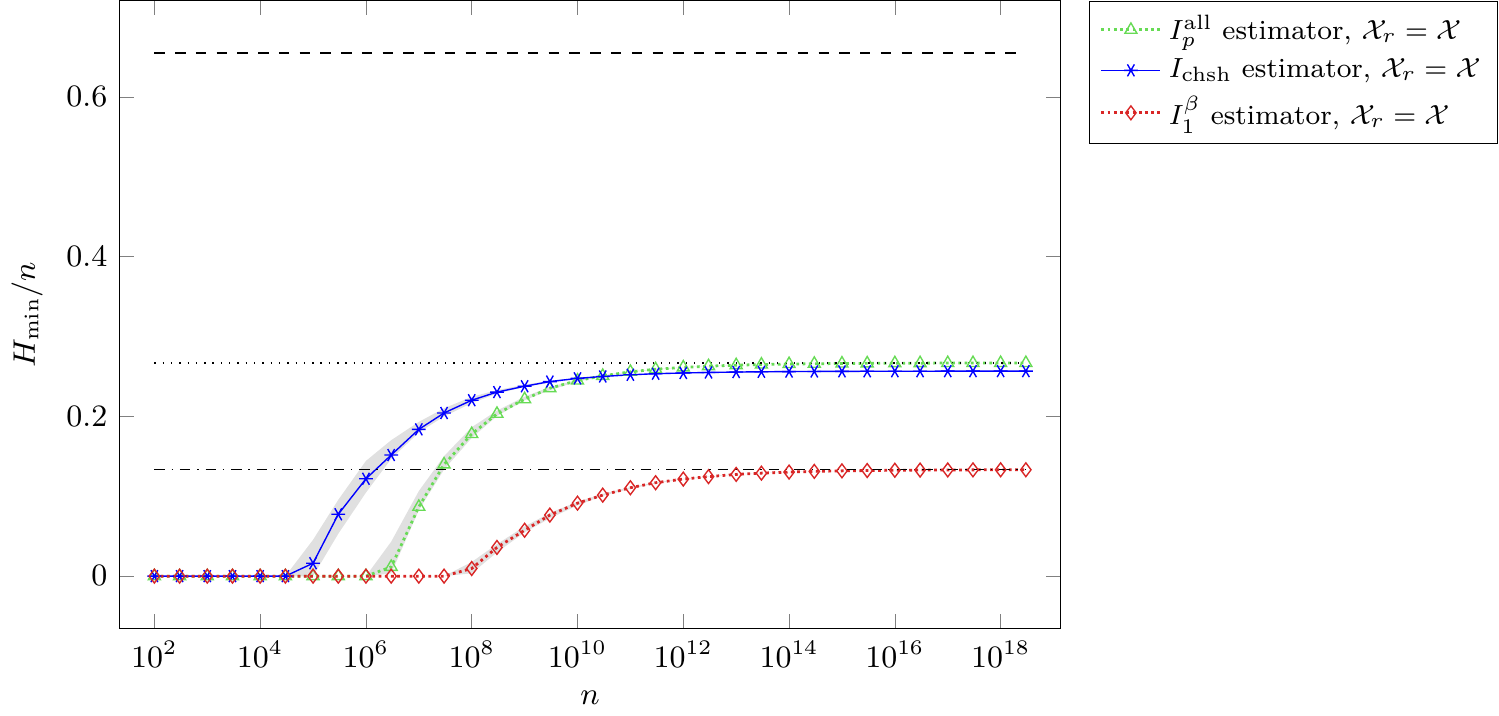}
\caption{Bounds on min-entropy rate $H_{\min}/n$ (eq.~\eqref{eq:Gf-x}) with $\mathcal X_r = \mathcal X$, for a varying number of measurement runs $n$ and for different Bell expressions.
Note that for $\mathcal X_r = \mathcal X$, the error term $\nu(\vec x)\eta/n$ equals zero.
The protocol is simulated for the device behavior $p$ in eq.~\eqref{eq:behavior} as explained in Section~\ref{sec:pra_case}.
The dashed and dotted lines represent the min-entropy $H(p)$ associated to the behavior $p$, respectively for $\mathcal X_r = \{(1,0)\}$ and $\mathcal X_r = \mathcal X$.
The dash-dotted line represents the min-entropy $H(I_1^\beta[p])$ for $\mathcal X_r = \mathcal X$ given the expectation value of the tilted-CHSH expression $I_1^\beta$.
}
\label{fig:1}
\end{figure}

As we can see, the expression $I_1^\beta$ gives the worst results.
The reason for this is that the inequality is suited to the extremal behavior $p_\text{ext}$ rather than the imperfect behavior we simulated (i.e., in equation \eqref{eq:behavior}, the case of perfect visibility $v = 1$ rather than $v = 0.99$).
On the other hand, the expression $I_p^\text{all}$ is tailored to our illustrative  behavior \eqref{eq:behavior}; it thus gives asymptotically optimal results for $\mathcal X_r = \mathcal X$ according to \cite{NPS14,BSS14}.
There is, however, no reason for it to be optimal for finite, low values of $n$.
Indeed, we observe that the CHSH expression, while not specially suited to the behavior of our device, yields a better performance for values of $n$ lower than $10^{10}$.
The CHSH expression appears actually to be a good randomness certificate for all values of $n$, as it only performs slightly worse than $I_{p}^\text{all}$ asymptotically.

\subsection{Bounding randomness for a subset of all inputs (\texorpdfstring{$\mathcal{X}_r \subseteq \mathcal{X}$}{X\textrinferior{} \textsubseteq{} X})} \label{sec:lessinputs}
Having reviewed the case $t=1$ and $\mathcal{X}_r=\mathcal{X}$, we proceed to consider the modifications introduced in this work.
We consider first the possibility $\mathcal{X}_r\subset \mathcal{X}$.
This means that the RB function $H$ is only required to non-trivially bound the output probability for inputs that are in the set $\mathcal{X}_r$.
This is an important feature because for many Bell expressions the randomness that can be certified depends on the  input used.
For instance, maximal violation of the tilted-CHSH inequalities may imply that the randomness is maximal for one input pair but near zero for another input pair \cite{AMP12}.
Using a function which is simultaneously randomness-bounding for all inputs $x\in\mathcal{X}$ would then be highly sub-optimal in this case.
This aspect is particularly important for photonic implementations of DI protocols: recent photonic Bell tests rely on partially entangled states \cite{CMA+13,GMR+13,GVW+15,SMC+15}, for which the optimal extraction of randomness requires the use of a specific input.

According to our analysis, in the case $\mathcal{X}_r\subseteq \mathcal{X}$, the bound \eqref{eq:Gf-} becomes
\begin{align}\label{eq:Gf-x}
H_\text{min} \gtrsim nH\mleft( \hat{f} - \gamma \sqrt{\frac{2}{n}\ln \frac{1}{\epsilon}} \mright) - \nu(\vec{x}) \eta \,,
\end{align}
where $H$ is now a RB function for $\mathcal{X}_r$, which will generally yield an improvement over a RB function that is required to be valid for all of $\mathcal{X}$.
Our analysis, however, introduces a penalty term of the form $\nu(\vec{x})\eta$, where $\eta \le \log_2\abs{\mathcal{A}}$ is bounded by a constant and $\nu(\vec{x})$ is the number of inputs not in $\mathcal{X}_r$ that have been observed.%
\footnote{Note that in certain cases, the RB function may be insensitive to the choice of $\mathcal{X}_r$.
For instance, in the case of the CHSH inequality, which is highly symmetric, putting all inputs on the same footing, any RB function returns the same bound on the randomness independently of the choice of $\mathcal{X}_r$, and in particular when $\mathcal{X}_r=\mathcal{X}$.
Thus the choice of $\mathcal{X}_r$ will not impact the main term of \eqref{eq:Gf-x} but only the penalty term, which will vanish if $\abs{\mathcal X_{\bar r}}=0$.
In such situations, it is therefore preferable to take $\mathcal{X}_r=\mathcal{X}$, as in \cite{PAM+10,PM13}.}
To keep this penalty term as low as possible, we should choose inputs in $\mathcal{X}_{\bar r}=\mathcal{X}\setminus\mathcal{X}_r$ with a low probability.
One possibility, compatible with our previous discussion about the introduction of a bias in the input distribution, is to take $\pi(x)=\kappa' n^{-\delta}$ for $x\in\mathcal{X}_{\bar r}$, in which case the expected value of $\nu(\vec{x})$ would be $\abs{\mathcal{X}_{\bar r}}\kappa'n^{1-\delta}$.
This is negligible asymptotically with respect to the main term of \eqref{eq:Gf-x}, which is $\Omega(n)$, provided that $\delta>0$.
The input distribution \eqref{eq:input-dist} chosen for our numerical example satisfies this requirement.

The corresponding min-entropy rate bound (that is, \eqref{eq:Gf-x} divided by $n$) for $\epsilon=10^{-6}$ and $\eta = \log_2\abs{\mathcal A} = 2$ is represented in Figure~\ref{fig:2} as a function of the number of runs $n$, for different Bell expressions and for two choices of $\mathcal X_r$.

\begin{figure}[t]
\centering
\includegraphics{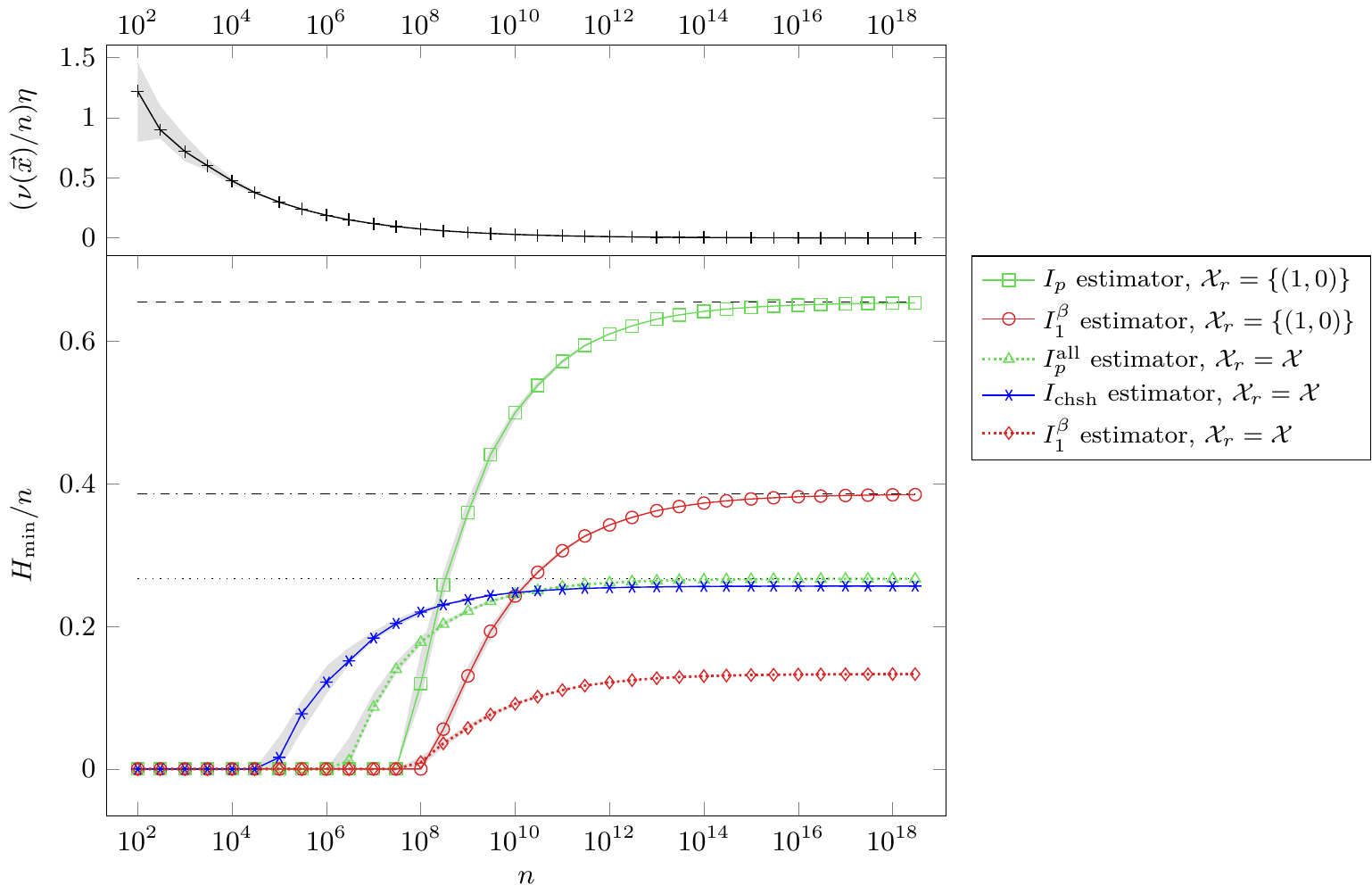}
\caption{%
(Top)
Penalty term in the min-entropy rate bound for different number of measurement rounds $n$.
(Bottom)
Bounds on min-entropy rate $H_{\min}/n$ (eq.~\eqref{eq:Gf-x}) for a varying number of measurement runs $n$, for different Bell expressions and for different input subsets $\mathcal X_r$.
The protocol is simulated for the device behavior $p$ in eq.~\eqref{eq:behavior} as explained in Section~\ref{sec:pra_case}.
The dashed and dotted lines represent the min-entropy $H(p)$ associated to the behavior $p$, respectively for $\mathcal X_r = \{(1,0)\}$ and $\mathcal X_r = \mathcal X$.
The dash-dotted line represents the min-entropy $H(I_1^\beta[p])$ for $\mathcal X_r = \{(1,0)\}$ given the expectation value of the tilted-CHSH expression $I_1^\beta$.}
\label{fig:2}
\end{figure}

Figure~\ref{fig:2} shows that in spite of the penalty term, which quickly vanishes as $n$ grows, the bound on the entropy rate for $\mathcal X_r = \{(1,0)\}$ for the expressions $I_1^\beta$ and $I_p$ (the analogue of $I_p^\text{all}$ for this restricted $\mathcal X_r$) is significantly better than the values obtained with $\mathcal X_r = \mathcal X$.
This clearly shows the value of using the restricted randomness-generating input set $\mathcal X_r=\{(1,0)\}$.
In particular, by combining this with the use of the optimal expression $I_p$ corresponding to behavior \eqref{eq:behavior}, one can asymptotically reach the theoretical value $H(p)$ of the min-entropy (represented by the dashed line in Figure~\ref{fig:2}).
Furthermore, whereas using the $I_1^\beta$ expression with $\mathcal X_r = \mathcal X$ yields worse values than using the CHSH expression, taking $\mathcal X_r = \{(1,0)\}$ yields an asymptotic entropy rate for the $I_1^\beta$ expression which is higher than using the CHSH expression.
As mentioned at the beginning of this subsection, this is because $I_1^\beta$ is not adapted to bound randomness independently of the input (i.e., $\mathcal X_r = \mathcal X$) \citep{AMP12}.

Similarly to the difference between the curves corresponding to $I_p^\text{all}$ and $I_1^\beta$ for $\mathcal X_r = \mathcal X$, the difference between the asymptotic entropy rates reached by using the expressions $I_p$ and $I_1^\beta$ for $\mathcal X_r = \{(1,0)\}$ is caused by the imperfect visibility parameter $v = 0.99$ in the simulated behavior \eqref{eq:behavior}.

Making the right choice of Bell expression and input subset $\mathcal X_r$ depends not only on the device, but also on the value of $n$.
Indeed, while $I_p$ is an optimal expression for certifying randomness in this specific device with respect to the input subset $\mathcal X_r = \{(1,0)\}$, this is only the case asymptotically.
For small $n$, Figure~\ref{fig:2} suggests that the CHSH expression has a better resistance to statistical fluctuations than the other expressions we considered, regardless of $\mathcal X_r$.

Note that we did not attempt to optimize the choice of input distribution and it is possible that a different choice of $\pi(x)$ would lead to better bounds in Figure~\ref{fig:2} for the two curves with $\mathcal X_r = \{(1,0)\}$.

\subsection{Bounding randomness from several Bell expressions (\texorpdfstring{$t \ge 1$}{t \textge{} 1})} \label{sec:multibell}
As we have seen, the right choice of a single Bell expression in the analysis of randomness is not straightforward, except for large values of $n$ where $I_p$ becomes optimal.
In this regime, it would seem perfectly admissible to perform tests on the device \emph{before} running the actual randomness generation protocol, in order to estimate $p$ and use this information to find an ``optimal'' Bell expression $I_{p'}$ as described above, which can afterwards be used in the randomness generation protocol proper.
However, there are disadvantages to this method.

Firstly, to find an expression $I_{p'}$ that performs comparably to the optimal $I_p$ for the device behavior $p$, we must know $p$ to a sufficiently high accuracy.
In a black box scenario where imperfections cannot be ruled out, this means that a significant number of measurements must be performed in order to evaluate the behavior to great precision.
Since the Bell expression needs to be fixed in advance of the protocol, those evaluation rounds cannot be taken from measurement rounds of the protocol and must instead be thrown away.
In addition, the behavior of the devices may vary in time, unlike our i.i.d.\ choice \eqref{eq:behavior}, due to drifts in the experimental set-up for example.
In that case, one would need to periodically estimate $p$ and rederive the corresponding optimal Bell expression on some subset of the measurement data that needs to be thrown away.
Finding an expression $I_{p'}$ also requires methods of inference of the behavior of the device from a finite sample: indeed, the estimated behavior \eqref{eq:freq} cannot be used directly to find a candidate $I_{\hat p}$, as $\hat p$ almost always violates the no-signaling conditions.
There exist different approaches to this inference (see for instance \cite{BSS14,MSA+17,SBSL16,LRZ+17,GMG+17}), so a nontrivial choice must be made.

Finally, even ignoring the problem of estimating the unknown behavior $p$, the associated data loss, or the drift of $p$ over time, we saw in the previous section and in Figure~\ref{fig:2} that the choice of a Bell expression is not straightforward when considering different values of $n$.
For example, the asymptotically optimal expression $I_p$ as formulated in \cite{NPS14,BSS14} is generally not the best for low values of $n$.
There is thus no general method to guide the choice of a Bell expression for a given $n$.

In order to avoid the above problems associated to the use of a single Bell expression,%
\footnote{%
Note that another issue is the choice of a randomness-generating input set $\mathcal X_r$.
The set $\mathcal X_r$ maximizing the randomness generation rate depends obviously on the underlying behavior $p$, but also on the number of rounds $n$, as illustrated in Fig.~\ref{fig:2}.
In practice, setting $\mathcal X_r$ could be the result of an informed choice based on prior information about the behavior of the devices or a rough estimate of it made before running the protocol.
It is reasonable to expect, as we do in our numerical simulation, that the optimal set $\mathcal X_r$ is only weakly sensitive to fluctuations or drifts of the experimental set-up, since it takes its value from a discrete set.
The use of a fixed set $\mathcal X_r$ is thus less problematic than the use of a fixed Bell estimator.
}
we now turn to the second element introduced in this work: the possibility to estimate the randomness from $t>1$ Bell expressions, and in particular from the full set of observed frequencies of occurrence $\hat{p} = \lbrace \hat{p}(a \given x) \rbrace$ as defined in eq.~\eqref{eq:freq}.

When we have more than one Bell expression, the bound \eqref{eq:Gf-x} generalizes to
\begin{align}\label{eq:GVx}
H_\text{min} \gtrsim nH\mleft( [\hat{\bff}^-,\hat{\bff}^+] \mright) - \nu(\vec{x})\eta\,,
\end{align}
where the one-dimensional interval $[\hat f^-,\infty]$ has simply been replaced with the multidimensional region $[\hat{\bff}^-,\hat{\bff}^+]$.
As before the limits of the region depend on the security parameter $\epsilon$, the constants $\gamma_\alpha$, and they become smaller with the number of runs $n$ (see equations \eqref{eq:muxvpm} and \eqref{eq:confreg}).

Increasing the number of Bell expressions can have both beneficial and detrimental consequences.
We can reach an understanding of this by considering the optimization problem \eqref{eq:guess5} that defines $H( [\hat{\bff}^-,\hat{\bff}^+] )=-\log_2 G( [\hat{\bff}^-,\hat{\bff}^+] )$.
This problem essentially evaluates the randomness of a certain quantum behavior $p$ such that $\hat{\bff}^- \le \bvec f[p] \le \hat{\bff}^+$.
Each vector component of this constraint defines two affine constraints, $\hat f_\alpha^- \le f_\alpha[p] \le \hat{f}_\alpha^+$, restricting the set of values $p$ can take in the optimization.
From a geometrical point of view, for each $\alpha = 1, \dotsc, t$, this defines two parallel hyperplanes in the space of behaviors, delimiting a region between them which we call a slab.
The full constraint $\hat{\bff}^- \le \bvec f[p] \le \hat{\bff}^+$ defines a polytope in the space of behaviors which is the intersection of the $t$ slabs.

The optimization \eqref{eq:guess5} identifies the worst-case bound on randomness for quantum behaviors inside this constraint polytope.
We would therefore like to restrict this region as much as possible given a value of the confidence parameter $\epsilon$.
We thus see that adding Bell expressions is generally beneficial, as it cuts the constraint polytope into a smaller volume.
However, as stated in Lemma~\ref{thm:azuma}, the confidence parameter $\epsilon$ in the protocol is shared between all $2t$ parameters $\epsilon^\pm_\alpha$, as $\sum_{\alpha=1}^t (\epsilon_\alpha^+ + \epsilon_\alpha^-) = \epsilon$.
A consequence of this is that the more Bell expressions we have, the smaller $\epsilon_\alpha^\pm$ are on average.
Since smaller values of $\epsilon^\pm_\alpha$ give thicker slabs (see equation~\eqref{eq:muxvpm}), if we are distributing $\epsilon$ evenly across all $\epsilon^\pm_\alpha$ for instance, this amounts to a dilation of the constraint polytope in optimization \eqref{eq:guess5}.
Nevertheless, since the width of the slab depends on $\epsilon^\pm_\alpha$ only through a factor $\sqrt{\ln(1/\epsilon^\pm_\alpha)}$, we will typically find that this negative effect is outweighed by the positives of adding more Bell expressions, and the randomness bound is globally improved.

\begin{figure}[tp]
\centering
\includegraphics{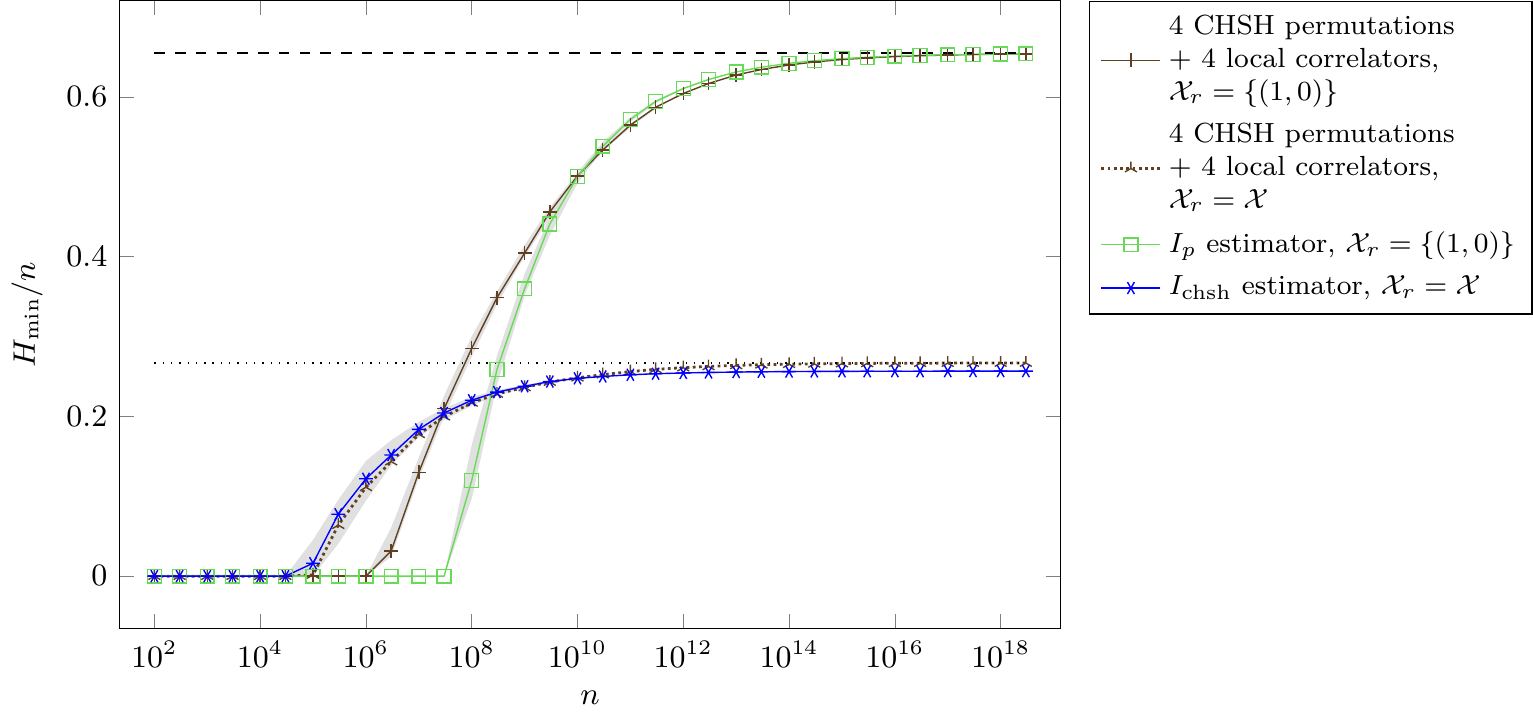}
\caption{Bounds on the min-entropy rate $H_{\min}/n$ (eq.~\eqref{eq:GVx}) for a single Bell expression ($I_\rchsh$ or $I_p$), and for a ``complete'' set of $8$ expressions.
The protocol is simulated for the device behavior $p$ in eq.~\eqref{eq:behavior} as explained in Section~\ref{sec:pra_case}.
The dashed and dotted lines represent the min-entropy $H(p)$ associated to the behavior $p$, respectively for $\mathcal X_r = \{(1,0)\}$ and $\mathcal X_r = \mathcal X$.
}
\label{fig:4}
\end{figure}

This improvement is illustrated in Figure~\ref{fig:4}, where we reconsider the numerical example presented in Section~\ref{sec:pra_case}.
We now use $8$ Bell expressions that are equivalent (for quantum behaviors) to the specification of the $16$ probabilities $p(\mathrm{a}_1\mathrm{a}_2 \given \mathrm{x}_1\mathrm{x}_2)$ (we will explain why we use this choice of $8$ Bell expressions later).

In the case $\mathcal X_r = \{(1,0)\}$, we see that the randomness that we can extract with multiple Bell expressions is similar to the use of the optimal expression $I_p$ alone for about $n\geq 10^{10}$ runs, but much better for smaller numbers of runs.
In fact, the improvement is even better in practice, even in regions where the use of $8$ Bell expressions gives the same rate as $I_p$, because in plotting the rate for $I_p$ we knew the \emph{exact} behavior $p$ from which the measurement runs are sampled.
In a real experiment (in particular with drifts over time),
we would instead need to infer the behavior $p$ at regular intervals from auxiliary measurements and throw away the corresponding data.
In contrast, the method based on the full set of observed frequencies achieves the same or better randomness extraction without throwing away any data.

In the case $\mathcal X_r = \mathcal X$, the use of $8$ expressions is comparable to that of CHSH
and, as with $I_p$ above, it outperforms the $I_p^\text{all}$ expression (not shown in Figure~\ref{fig:4}; see Figure~\ref{fig:1} or \ref{fig:2}).
For small values of $n$, the CHSH expression keeps a small advantage.
This can be understood from the fact that the CHSH expression itself is part of the set of $8$ expressions that we used in Figure~\ref{fig:4}.
The difference between the two curves therefore results from a trade-off between a better estimation of randomnes from more expressions and the negative effect of wider margins $\epsilon^\pm_\alpha$ in the confidence region.

In the remainder of this section, we discuss in more detail how to choose a good set of Bell expressions for the protocol.
For this, let us start by considering our protocol when $n\rightarrow\infty$.
In this asymptotic limit, the interval $[\hat{\bff}^-,\hat{\bff}^+]$ narrows down towards the point $\hat{\bff}=\bff[\hat{p}(a \given x)]$, which is just the value of the $t$ Bell expressions $\bff$ computed on the experimentally observed frequencies $\hat{p}(a \given x)$.
If the bias towards inputs in $\mathcal{X}_r$ is appropriately chosen (as discussed previously), then the relative contribution of the penalty term vanishes as $n \rightarrow \infty$ and the bound \eqref{eq:GVx} becomes in the asymptotic limit, up to sublinear terms,
\begin{align}
H_\text{min} \gtrsim nH( \hat{\bff} ) \,.
\end{align}
Furthermore, in the case where the device behaves in an i.i.d.\ way according to a behavior $p$, then, asymptotically, $\hat{\bff}\rightarrow \bff[p]$.
If one chooses enough Bell expressions as to fully characterize the behavior of the devices (for instance, by using an estimator for each probability $p(a \given x)$),  $\hat{\bff}$ thus becomes equivalent to the knowledge of $p$ and the above bound converges to the maximal min-entropy bound one can obtain from $p$ given $\mathcal X_r$, as characterized in \cite{NPS14,BSS14}.
In this sense, and as seen in Figure~\ref{fig:4}, our protocol is asymptotically optimal.

Note that there are different sets of Bell estimators that are asymptotically equivalent to the knowledge of the full set of probabilities $p(a \given x)$.
For instance in a bipartite Bell experiment with two inputs and two outputs there are $16$ probabilities $p(a \given x)=p(\mathrm{a}_1\mathrm{a_2} \given \mathrm{x}_1\mathrm{x}_2)$ with $\mathrm{a}_1,\mathrm{a}_2,\mathrm{x}_1,\mathrm{x}_2\in\{0,1\}$ and thus $16$ associated Bell expressions $e_1,\dotsc,e_{16}$ defined by $e_\alpha[p]=p(\mathrm{a}_1\mathrm{a_2} \given \mathrm{x}_1\mathrm{x}_2)$, with one value of $\alpha$ for each of the possible values of $(\mathrm{a}_1,\mathrm{a}_2,\mathrm{x}_1,\mathrm{x}_2)$.
But since the probabilities $p(a \given x)$ satisfy normalization and no-signaling, they are uniquely specified by the $8$ correlators of eq.~\eqref{eq:correlators},
which constitute $8$ Bell expressions $g_1,\dotsc,g_8$, where $g_1$ and $g_2$ are the first party's two marginal correlators $\mean{A_{\mathrm x_1}}$, $g_3$ and $g_4$ are the second party's $\mean{B_{\mathrm x_2}}$, and $g_4, \dotsc, g_8$ are the four bipartite correlators $\mean{A_{\mathrm x_1} B_{\mathrm x_2}}$.

Alternatively, the probabilities are also equivalent to the $8$ expressions $h_1,\dotsc,h_8$ with $h_\alpha = g_\alpha$ for $\alpha = 1,\dotsc,4$, and $h_\alpha$ for $\alpha = 5,\dotsc,8$ are four linearly independent permutations of the CHSH expression, generalizing \eqref{eq:chsh}:
\begin{equation}
\label{eq:genchsh}
I_\rchsh^{\mathrm{y}_1,\mathrm{y}_2}
 = \sum_{\mathrm x_1, \mathrm x_2 \in \{0,1\}}
  (-1)^{(\mathrm{x}_1+\mathrm{y}_1)(\mathrm{x}_2+\mathrm{y}_2)}
  \mean{A_{\mathrm{x}_1} B_{\mathrm{x}_2}}
 \qquad \mathrm{y}_1,\mathrm{y}_2 \in\{0,1\}.
\end{equation}

As we increase the number of rounds, all these possible choices become equivalent, since the intervals $[\hat{\mathbf{e}}^-,\hat{\mathbf{e}}^+]$, $[\hat{\mathbf{g}}^-,\hat{\mathbf{g}}^+]$, $[\hat{\mathbf{h}}^-,\hat{\mathbf{h}}^+]$ define constraint polytopes in the space of behaviors $p$ that asymptotically intersect the quantum set at the same unique point.
However, the choice of one set of estimators over another could make a difference for finite $n$.

Generally speaking, when choosing which Bell expressions to use for a fixed number $t$, we may prefer that as many of them as possible be linearly independent.
Consider $t-1$ Bell expressions and their associated slabs, which define a constraint polytope.
In the absence of any meaningful information concerning the behavior of the objective function of \eqref{eq:guess5} within its feasible set, the choice of a $t$-th Bell expression should be dictated by the resulting reduction of the constraint polytope: cutting a large volume out is more likely to reduce the maximum of \eqref{eq:guess5}.
As $n$ grows large and the slabs grow thinner, the best way to reduce this volume is to choose a Bell expression that is linearly independent from the $t-1$ previous ones, if possible.
We can easily understand this in the asymptotic limit: as we mentioned above, the optimization converges to $H(\hat \bff)$, and with enough linearly independent Bell expressions, $\bff[\hat p]$ uniquely defines $\hat p$, hence $H(\hat \bff) = H(\hat p)$.
At this point, adding more expressions only makes $\bff[\hat p]$ a more redundant definition of $\hat p$, which does not improve the randomness bound.
On the other hand, with too few independent Bell expressions, $\hat \bff = \bff[\hat p]$ is compatible with many values of $\hat p$, and the worst value is what ends up determining $H(\hat \bff)$.

In addition, we see that there is no need for Bell expressions that are purely signaling, i.e., that have a constant value for all no-signaling behaviors $p$.
Indeed, since the feasible region of \eqref{eq:guess5} is defined by the intersection of the slabs and the quantum set, constraints deriving from purely signaling expressions are trivial in this region, and therefore do not contribute to improve the randomness bound.

Combining these two conclusions also indicates that we should avoid Bell expressions that are only linearly dependent up to purely signaling terms.
This implies for instance that the sets $g_1,\dotsc,g_8$ or $h_1,\dotsc,h_8$ should be preferred over $e_1,\dotsc,e_{16}$.
This is indeed what we find, as illustrated in Figure~\ref{fig:5}.

Note that the sets $g_1,\dotsc,g_8$ and $h_1,\dotsc,h_8$ only differ by a linear transformation, but the second set yields better results for the same (finite) number of rounds $n$.
With respect to optimization \eqref{eq:guess5}, this means that the feasible set for $h_1,\dotsc,h_8$ excluded the optimum obtained for $g_1,\dotsc,g_8$.
This might be related to the fact that in this scenario of two parties with two inputs and two outputs, the four versions of the CHSH inequalities constitute the facets that separate local from nonlocal behaviors, and they might therefore serve as better measures of nonlocality and randomness than the correlators $\mean{A_\mathrm{x_1}B_\mathrm{x_2}}$.

\begin{figure}[tp]
\centering
\includegraphics{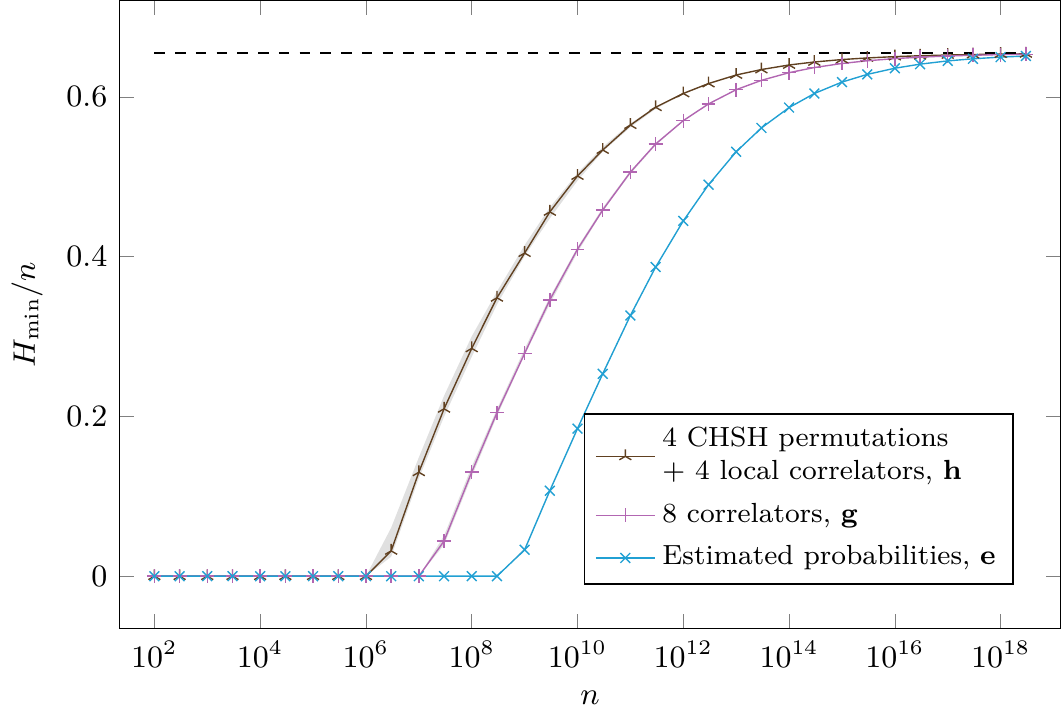}
\caption{%
Bounds on the min-entropy rate $H_{\min}/n$ (eq.~\eqref{eq:GVx}) with $\mathcal X_r = \{(1,0)\}$ for three (asymptotically) equivalent sets of Bell expressions.
The protocol is simulated for the device behavior $p$ in eq.~\eqref{eq:behavior} as explained in Section~\ref{sec:pra_case}.
The dashed line represents the min-entropy $H(p)$ associated to the behavior $p$ for $\mathcal X_r = \{(1,0)\}$.
}
\label{fig:5}
\end{figure}

This phenomenon can be visualized in a simpler instance where two Bell expressions are used, namely, $I_\rchsh^{0,0}$ and $I_\rchsh^{0,1}$.
In Figure~\ref{fig:heatplot}, we represent the RB function $G(I_\rchsh^{0,0}[p], I_\rchsh^{0,1}[p])$ with respect to a single input $\mathcal X_r = \{(0,0)\}$, as defined in eq.~\eqref{eq:guess1b}.
The figure shows the evolution of the randomness bound, from a trivial value of $1$ for values of $I_\rchsh^{0,0}[p]$ and $I_\rchsh^{0,1}[p]$ compatible with a local hidden variable model, to nontrivial values up to approximately $0.32$ reached at extremal points.
The variation of the RB function along these two axes is not trivial, but as we can see, the gradient mostly points along the directions of the CHSH axes, with the exception of the regions where the local and quantum boundaries meet.
To minimize this variation within a confidence region of fixed square shape in this plane, it is best to rotate the interval so that its sides are aligned with the gradient.
Aligning the confidence region with the two CHSH axes is therefore a sensible choice in this case.

\begin{figure}[t]
\centering
\includegraphics{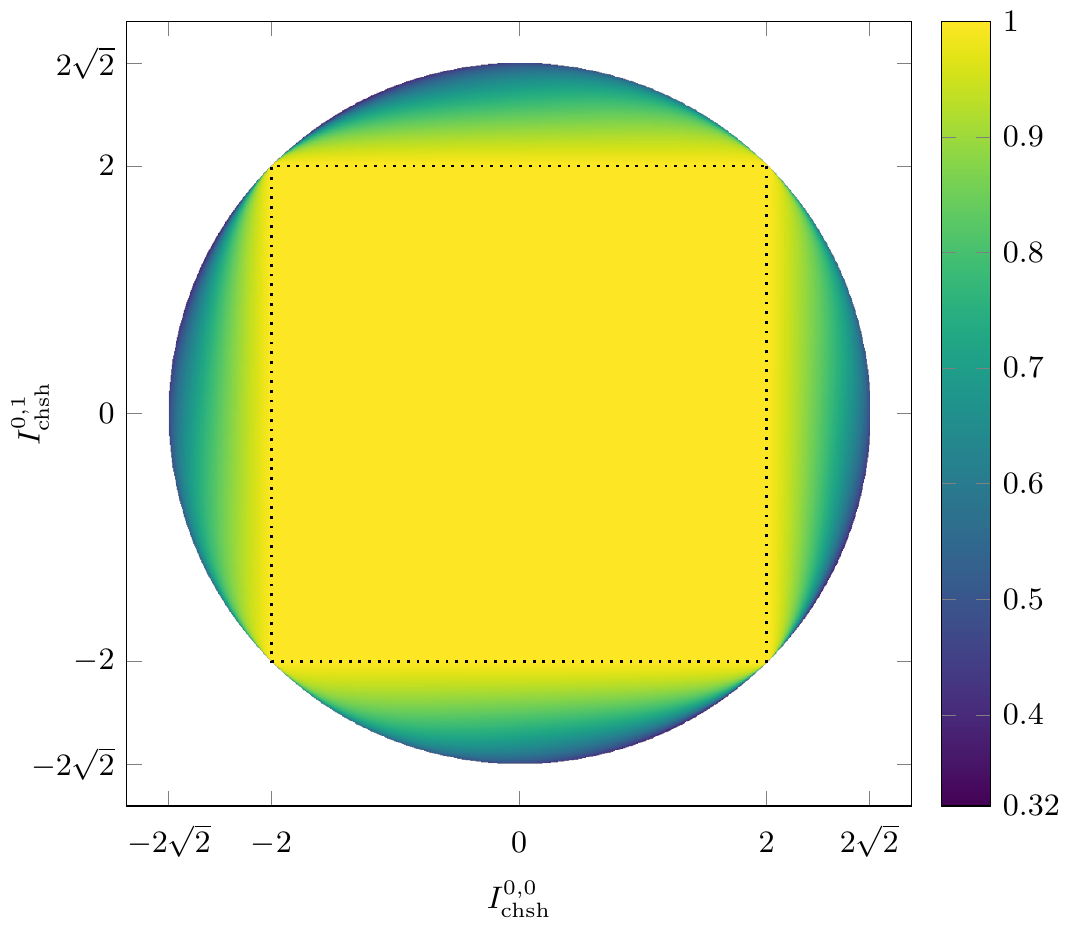}
\caption{%
RB function $G(f_1, f_2)$ according to \eqref{eq:guess1b} for $\mathcal X_r = \{(0,0)\}$, with $f_1 = I_\rchsh^{0,0}$ and $f_2 = I_\rchsh^{0,1}$, as defined in \eqref{eq:genchsh}.
The set of local behaviors projected onto this plane defines a square region, represented with a dotted line.
The quantum set $\mathcal Q$ projects into the circular region $(I_\rchsh^{0,0}[p])^2 + (I_\rchsh^{0,1}[p])^2 \le 8$ \cite{Mas03}.
Note that while this RB function is centrally symmetric, it is not symmetric under reflection through either CHSH axis.
The function was computed by solving the optimization program \eqref{eq:guess1b} at level $2$ of the NPA hierarchy.
}
\label{fig:heatplot}
\end{figure}

\section{Conclusion and open questions}
\label{sec:conclusion}
In recent years several protocols for generating randomness in a DI way have been introduced \cite{PAM+10,VV12,RUV13,CY13,ARV16,MS14,PM13,FGS13,MS14a,Col09,*CK11,MY98,ABG+07,MRC+14,Mas09,MPA11,PMLA13}, with varying degrees of security, rate of randomness expansion, or noise robustness.
They all, however, share the feature that they rely on the estimation of a single Bell expression.

As was shown in \cite{NPS14,BSS14}, in an idealized setting in which the behavior of the devices is known and fixed, more randomness can in principle be certified if one takes into account the violation of several Bell inequalities, or, even better, the full set of probabilities characterizing the devices' behavior.

We have shown here that a similar reasoning applies in the context of an actual DIRNG protocol, where randomness is directly certified from experimental data.
Specifically, we have combined the analysis of \cite{NPS14,BSS14} with the protocol introduced in \cite{PAM+10,PM13,FGS13}, which generates certified randomness against an adversary with classical side information.
We have in this way obtained a family of DIRNG protocols which rely on the estimation of a choice of $t\geq 1$ Bell expressions.
This includes the special case where the randomness is directly certified from the knowledge of the relative frequencies of occurrence of the outputs given the inputs.
Asymptotically, for a given $\mathcal X_r$, this results in an optimal generation of randomness from experimental data (as measured by the min-entropy) without having to assume beforehand that the devices violate a specific Bell inequality and without the need to infer the device behavior from preliminary measurements.
Furthermore, in the non-asymptotic case, the choice of an optimal Bell expression is ambiguous even if the device behavior is perfectly characterized.
Our method proposes a way of bypassing this problem by directly evaluating the randomness from the observed output frequencies.

Our protocol also provides a way of treating the case where the randomness of the outcomes of the devices is much higher for some inputs than for others.
This happens in particular when generating randomness from partially entangled states \cite{AMP12}, which are used in present photonic loophole-free Bell experiments \cite{GVW+15,SMC+15}.
Our analysis essentially amounts to consider that all the randomness has been generated from the optimal set of inputs, but corrected by a penalty term that is proportional to the number of events corresponding to non-optimal inputs.
By biasing the choices of inputs towards the optimal ones, one can make this penalty term negligible asymptotically.
However, for small numbers of measurement runs, we have seen that this procedure may be less efficient than an analysis based on a Bell expression that treats all inputs on the same footing, like the CHSH expression. It is possible that our way of treating non-optimal inputs could be improved, leading to more efficient protocols for small numbers of measurement runs.

Our result could be generalized in several ways.
First, how to prove security against quantum side information when several Bell expressions or the full set of data generated in the experiment are taken into account?
This is not a priori easy to answer since the analysis of most DIRNG protocols secure against quantum side information  rely on Bell expressions with a particular structure \cite{VV12,CY13,MS14} or, when they allow for arbitrary Bell expressions, do not optimally take into account the observed level of violation \cite{MS14a}.
Second, we based the statistical analysis on the Azuma-Hoeffding inequality, but alternative deviation theorems \cite{EW15} could be adapted to our setting.
Our attempt at improving our bounds using a tighter concentration inequality from Hoeffding \cite{Hoe63} (called McDiarmid's inequality in \cite{EW15} after \cite{McD89}) produced no visible difference in the plots.
Another alternative, Bentkus' inequality \cite{Ben04,EW15}, involves discrete summations with around $n$ terms, which would grow too large in our simulations to be used as-is.
Finally, we note that since DIRNG is not the only task where information from several Bell estimators can be exploited, a similar approach could be developed for other DI problems, such as DI quantum key distribution.

\section*{Acknowledgments}
We would like to thank the anonymous referees for the time dedicated to reviewing our manuscript and for their useful comments.

We acknowledge financial support from the Fondation Wiener-Anspach, the Interuniversity Attraction Poles
program of the Belgian Science Policy Office under the grant IAP P7-35 photonics@be, and the Fonds de la Recherche Scientifique F.R.S.-FNRS (Belgium).
O. N.S. acknowledges financial support through a grant of the Fonds pour la Formation \`{a} la Recherche dans l'Industrie et l'Agriculture (F.R.I.A.).
C. B. acknowledges funding from the F.R.S.-FNRS through a Research Fellowship.
J. S. was a postdoctoral researcher of the F.R.S.-FNRS at the time this research was carried out.
S. P. is a Research Associate of the F.R.S.-FNRS.
All plots were produced using TikZ \cite{tikz} and \textsc{pgfplots} \cite{pgfplots}.
The heat map of Figure~\ref{fig:heatplot} was produced using Matplotlib \cite{matplotlib}.

\appendix
\section{Appendix}
The problems introduced in Section~\ref{sec:sdp} are typical instances of conic programming.
A conic program is an optimization problem of the form
\begin{equation}
\begin{split}\label{eq:coneprog}
\maximize_{\{x_i\}} &\quad\sum_{i=1}^k \iprod{c_i,x_i}\\
\text{subject to} &\quad\sum_{i=1}^k A_i x_i = b\\
& \quad x_i\in K_i \qquad i=1,\dotsc,k \,,
\end{split}\end{equation}
where $b\in\mathbb{R}^m$, $x_i,c_i\in\mathbb{R}^{n_i}$, $A_i\in\mathbb{R}^{m\times n_i}$, and $K_i$ are closed convex cones ($i=1,\dotsc,k$).
Here $\iprod{c,x}$ denotes the usual scalar product $\iprod{c,x}=\sum_{j=1}^n c_jx_j$ between vectors in $\mathbb{R}^n$ and a set $K$ is a convex cone if $q_1 p_1 +q_2p_2$ belongs to $K$ for any nonnegative scalars $q_1,q_2\geq 0$ and any $p_1,p_2$ in $K$.

In other words, an optimization problem is a conic program if it involves the maximization of a linear function of the optimization variables given linear constraints on such variables and the condition that they belong to a certain family of cones.
Many familiar convex optimization problems, including linear and semidefinite programming, are instances of conic programming.

The problems \eqref{eq:guess4} and \eqref{eq:guess5}, and their NPA relaxations, are also conic programs.
Indeed, the set of unnormalized behaviors $\tilde{\mathcal{Q}}$ and all of its NPA relaxations are clearly cones and all other constraints in \eqref{eq:guess4} and \eqref{eq:guess5} are linear in the variables $\tilde p_{a,x}$.
As we said earlier, the NPA relaxations of these problems actually correspond to a specific type of conic programs, i.e., semidefinite programs.

Any conic program in the primal form \eqref{eq:coneprog} admits a dual formulation given by \cite{BV04}
\begin{equation}\begin{split}\label{eq:conedual}
\minimize_{y} &\quad \iprod{b,y}\\
\text{subject to} &\quad A_i^Ty-c_i \in K_i^* \qquad i=1,\dotsc,k
\end{split}\end{equation}
where $y\in\mathbb{R}^m$ and $K_i^*=\{z_i\in\mathbb{R}^{n_i}:\iprod{z_i,x_i} \geq 0, \,\forall x_i\in K_i\}$ is the dual cone of $K_i$.
Let $\alpha$ denote the value of the primal program \eqref{eq:coneprog} and $\beta$ the value of the dual program \eqref{eq:conedual}.
Then the strong duality theorem of conic programming says that if the primal program is feasible, has finite value $\alpha$, and has points $\tilde x_i\in\operatorname{int}(K_i)$ such that $\sum_{i=1}^k A_i \tilde x_i = b$, then the dual program is feasible and has finite value $\beta=\alpha$ \cite{BV04}.

Let us determine the dual of the optimization problem \eqref{eq:guess4}.
From now on, we take generically $\tilde{\mathcal{Q}}$ to denote the cone of unnormalized quantum behaviors or any of its NPA relaxations, as the analysis is identical in both cases.

Let us first rewrite \eqref{eq:guess4} in a form similar to \eqref{eq:coneprog}.
For this, define $c_{ax}\in\mathbb{R}^{\abs{\mathcal A}\times \abs{\mathcal X}}$ as the vector with components $c_{ax}(a',x')=\delta_{aa'}\delta_{xx'}$.
Thus $\iprod{c_{ax},p} = p(a \given x)$.
Let $b=(1,f_1[p],\dotsc,f_t[p]) \in \mathbb{R}^{1+t}$.
Let $u$ be any Bell expression such that $u[p] = \tr[p]$ for all no-signaling behaviors $p$, for instance $u(a,x) = \delta_{x,x_0}$ for some input $x_0 \in \mathcal X$.
Let $A_{ax}$ be matrices in $\mathbb{R}^{1+t,\abs{\mathcal A}\times \abs{\mathcal{X}}}$ with components $[A_{ax}]_{1,a'x'} = u(a',x')$ and $[A_{ax}]_{1+\alpha,a'x'}=f_\alpha(a',x')$ for $\alpha=1,\dotsc,t$.
Then
\begin{equation}\begin{split}\label{eq:guess_cone}
G(\bvec{f}[p])  = \max_{\{\tilde p_{a,x}\}_{a\in\mathcal{A},x\in\mathcal{X}_r}} &\quad \sum_{a\in\mathcal{A},x\in\mathcal{X}_r} \iprod{c_{ax},\tilde p_{a,x}} \\
\text{subject to} & \quad \sum_{a\in\mathcal{A},x\in\mathcal{X}_r} A_{ax} \tilde p_{a,x}= b\\
 & \quad \tilde p_{a,x}\in \tilde{\mathcal{Q}} \quad \forall a\in\mathcal{A},\forall x\in \mathcal{X}_r.
\end{split}\end{equation}
The dual is then readily given as
\begin{equation}\begin{split}\label{eq:guess_cone_dual}
G(\bvec{f}[p])=\min_{y\in\mathbb{R}^{1+t}} &\quad \iprod{b,y}\\
\text{subject to} &\quad A_{ax}^Ty-c_{ax} \in \tilde{\mathcal Q}^* \qquad \forall a\in\mathcal{A},\forall x\in \mathcal{X}_r\,,
\end{split}\end{equation}
where $\tilde{\mathcal{Q}}^*$ is the dual cone of $\tilde{\mathcal{Q}}$, that is, the set $\tilde{\mathcal{Q}}^*=\{c\in\mathbb{R}^{\abs{\mathcal A}\times\abs{\mathcal X}}: \iprod{c,\tilde p}\geq 0,\,\forall \tilde p\in \tilde{\mathcal{Q}}\}$.
Note that this dual cone can be identified with the set of Tsirelson inequalities for normalized behaviors, that is, the set $\mathcal{Q}^*=\{(d_0,d)\in\mathbb{R}^{1+\abs{\mathcal A}\times\abs{\mathcal X}}: \iprod{d,p}\leq d_0,\,\forall p\in {\mathcal{Q}}\}$.
Indeed, note that by normalization $\iprod{u,p}=1$ and thus an inequality $\iprod{d,p}\leq d_0$ valid for $p\in\mathcal{Q}$ can always be rewritten in the form $\iprod{(d_0 u-d),p} \geq 0$, hence in the form $\iprod{c,p}\geq 0$ for some suitable $c$.
But an inequality $\iprod{c,p}\geq 0$ is clearly valid for $\mathcal{Q}$ if and only if it is valid for $\tilde{\mathcal{Q}}$.

Using the explicit form of $b$, $A_{ax}$, and $c_{ax}$ and the above interpretation of $\tilde{\mathcal{Q}^*}$, we can rewrite the dual \eqref{eq:guess_cone_dual} as
\begin{equation}\begin{split}\label{eq:g_v(p)}
G(\bvec{f}[p])=\min_{y\in\mathbb{R}^{1+t}} &  \quad y_0+\sum_{\alpha=1}^t y_\alpha f_\alpha[p]  \\
\text{subject to}&\quad p'(a \given x)\leq y_0+\sum_{\alpha=1}^t y_\alpha f_\alpha[p']\quad \forall a\in\mathcal{A},\, \forall x\in\mathcal{X}_r,\,\forall p'\in\mathcal{Q}.
\end{split}\end{equation}
The interpretation of this dual problem is immediate.
We are seeking an affine function $y_0+\sum_a y_\alpha f_\alpha[p]$ satisfying the condition~1 in Definition~\ref{def:guess}---this is the constraint line of the optimization problem \eqref{eq:g_v(p)}.
Since the function is affine, it is also concave and the condition~2 in Definition~\ref{def:guess} is also trivially satisfied.
We are then searching for the best function of this form, i.e., the one that gives the lowest possible value (i.e., the most randomness)  for the given Bell expectations $\bvec{f}[p]$.
Note that considering an affine function is actually not a restriction because the optimal function $G(\bvec{f}[p])$ is concave, and therefore it is equivalent to the envelope of its tangents, which are affine functions of $\bvec{f}[p]$.

From the point of view of implementations, the condition that $p'(a \given x)\leq y_0+\sum_\alpha y_\alpha f_\alpha[p']$ is valid for $\mathcal{Q}$ can be cast as a search for a sum-of-squares decomposition of this inequality.
When $\mathcal{Q}$ is one of the NPA relaxations of the quantum set, this is equivalent to solving a certain SDP, which turns out to be, as expected, the dual of the original NPA problem.

Let us now determine the dual of the noisy problem \eqref{eq:guess5}, which we can rewrite as
\begin{equation}\begin{split}\label{eq:noisy_conic}
G(\freg)  = \max_{\substack{\{\tilde p_{a,x}\}_{a\in\mathcal{A},x\in\mathcal{X}_r}\\\boldsymbol{\tau},\boldsymbol{\kappa}\in\mathbb{R}^{t}}} &\quad \sum_{a\in\mathcal{A},x\in\mathcal{X}_r} \tilde p_{a,x}(a \given x)\\
\text{subject to} & \quad \sum_{a\in\mathcal{A},x\in\mathcal{X}_r}  \tr[\tilde p_{a,x}]  =1\\
 &\quad \sum_{a\in\mathcal{A},x\in\mathcal{X}_r}\, \bvec{f}[\tilde p_{a,x}]+\boldsymbol{\tau}= \hat{\bff}^+\\
  &\quad \sum_{a\in\mathcal{A},x\in\mathcal{X}_r}\, \bvec{f}[\tilde p_{a,x}]-\boldsymbol{\kappa}=\hat{\bff}^-\\
& \quad \tilde p_{a,x}\in \mathcal{\tilde Q} \quad \forall a\in\mathcal{A},\,\forall x\in \mathcal{X}_r\\
&\quad \boldsymbol{\tau},\boldsymbol{\kappa}\in \mathbb{R}^t_+ \,.
\end{split}\end{equation}
Again, let us rewrite in a form similar to \eqref{eq:coneprog}.
For this, define $b\in\mathbb{R}^{1+2t}=(1,\hat{f}^+_1,\dotsc,\hat{f}^+_t,\hat{f}^-_1,\dotsc,\hat{f}^-_t)$.
Let $A_{ax}$ be matrices in $\mathbb{R}^{1+2t,\abs{\mathcal A}\times \abs{\mathcal{X}}}$ with components $[A_{ax}]_{1,a'x'} = u(a',x')$ and $[A_{ax}]_{1+\alpha,a'x'}=f_\alpha(a',x')$ as above, and $[A_{ax}]_{1+t+\alpha,a'x'}=f_\alpha(a',x')$, for $\alpha=1,\dotsc,t$.
Additionally, we define two matrices $A_1,A_2$ in $\mathbb{R}^{1+2t,t}$ through $[A_1]_{1,\alpha}=0$, $[A_1]_{1+\alpha,\alpha'}=\delta_{\alpha,\alpha'}$, $[A_1]_{1+t+\alpha,\alpha'}=0$ and $[A_2]_{1,\alpha}=0$, $[A_2]_{1+\alpha,\alpha'}=0$, $[A_2]_{1+t+\alpha,\alpha'}=-\delta_{\alpha,\alpha'}$ for $\alpha, \alpha' \in \{ 1, \dotsc, t \}$.
Then
\begin{equation}\begin{split}\label{eq:noisy_conic2}
G(\freg)  = \max_{\substack{\{\tilde p_{a,x}\}_{a\in\mathcal{A},x\in\mathcal{X}_r}\\\boldsymbol{\tau},\boldsymbol{\kappa}\in\mathbb{R}^{t}}}   &\quad \sum_{a\in\mathcal{A},x\in\mathcal{X}_r} \iprod{c_{ax},\tilde p_{a,x}}\\
\text{subject to} & \quad \sum_{a\in\mathcal{A},x\in\mathcal{X}_r} A_{ax} \tilde p_{a,x}+A_1\boldsymbol{\tau}+A_2\boldsymbol{\kappa}= b\\
 & \quad \tilde p_{a,x}\in \mathcal{\tilde Q} \quad \forall a\in\mathcal{A},\,\forall x\in \mathcal{X}_r\\
&\quad \boldsymbol{\tau},\boldsymbol{\kappa}\in \mathbb{R}^t_+ \,.
\end{split}\end{equation}
The dual is then readily given as
\begin{equation}\begin{split}\label{eq:noisy_conic_dual}
G(\freg) = \min_{v\in\mathbb{R}^{1+2t}} &\quad \iprod{b,v}\\
\text{subject to} &\quad A_{ax}^Tv-c_{ax} \in \tilde{\mathcal{Q}}^* \qquad \forall a\in\mathcal{A},\,\forall x\in \mathcal{X}_r\\
&\quad A_{1}^Tv,A_2^Tv \in \mathbb{R}^t_+
\end{split}\end{equation}
or, writing $v=(v_0,v_1,\dotsc,v_t, v'_1,\dotsc,v'_t)$
\begin{equation}\begin{split}\label{eq:noisy_conic_dual2}
G(\freg) = \min_{v\in\mathbb{R}^{1+2t}}   & \quad v_0+\sum_{\alpha=1}^t v_\alpha \hat{f}^+_\alpha+\sum_{\alpha=1}^t v'_\alpha \hat{f}^-_\alpha  \\
\text{subject to}&\quad v_0 u+\sum_{\alpha=1}^t (v_\alpha+v'_\alpha) {f}_\alpha-c_{ax}\in \tilde{\mathcal{Q}}^* \quad\forall a\in\mathcal{A},\,\forall x\in \mathcal{X}_r \\
&\quad v_\alpha \geq 0, v'_\alpha \leq 0 \quad \forall \alpha=1,\dotsc,t \,.
\end{split}\end{equation}
Defining $(y_0, y_\alpha^+, y_\alpha^-) = (v_0, v_\alpha, -v'_\alpha)$ and $y_\alpha=v_\alpha+v'_\alpha = y_\alpha^+ - y_\alpha^-$, we can reformulate this as
\begin{equation}\begin{split}\label{eq:noisy_conic_dual3}
G(\freg) =  \min_{(y_0,y^+,y^-)\in\mathbb{R}^{1+2t}}  & \quad y_0+\sum_{\alpha=1}^t y^+_\alpha \hat{f}^+_\alpha -\sum_{\alpha=1}^t y^-_\alpha \hat{f}^-_\alpha \\
\text{subject to}&\quad p'(a \given x)\leq y_0+\sum_{\alpha=1}^t y_\alpha f_\alpha[p']\quad \forall a\in\mathcal{A},\, \forall x\in\mathcal{X}_r,\,\forall p'\in\mathcal{Q}\\
&\quad y_\alpha=y^+_\alpha-y^-_\alpha\,,\quad y^+_\alpha,y^-_\alpha \geq 0\quad \forall \alpha=1,\dotsc,t \,.
\end{split}\end{equation}
This dual formulation is analogous to \eqref{eq:g_v(p)} and can be understood in the same way.
The difference lies in the objective function.

Although \eqref{eq:noisy_conic_dual3} is the best way to express the optimization for a numerical implementation, we may further simplify its expression by noting that, for a fixed value of $y_0$ and $\{y_\alpha\}$, the objective function is minimized when $y_\alpha^+ = y_\alpha$ if $y_\alpha \ge 0$ and $y_\alpha^- = -y_\alpha$ if $y_\alpha \le 0$ (in short, $y_\alpha^\pm = (\abs{y_\alpha} \pm y_\alpha)/2$).
We can then write the objective function as
\begin{equation}
	y_0 + \sum_{\alpha : y_\alpha \ge 0} y_\alpha \hat f_\alpha^+ + \sum_{\alpha : y_\alpha < 0} y_\alpha \hat f_\alpha^- \,,
\end{equation}
and further as the following maximum:
\begin{equation}
	\max_{\bvec f \in \freg} y_0 + \sum_{\alpha} y_\alpha f_\alpha \,.
\end{equation}
Indeed, the maximum will clearly be attained at one of the extreme points of the region $\freg$, whose components are of the form $f^\pm_\alpha$.
If $y_\alpha\geq 0$, the maximum will be attained when $f_\alpha$ is equal to $\hat{f}^+_\alpha$, if $y_\alpha\leq 0$, when it is equal to $\hat{f}^-_\alpha$.
All in all, \eqref{eq:noisy_conic_dual3} can thus be rewritten as the nested optimization
\begin{equation} \begin{split} \label{eq:minmaxdual}
	\min_{(y_0,y) \in \mathbb R^{1+t} \vphantom{\freg}} \, \max_{\bvec f \in \freg} &\quad y_0 + \sum_{\alpha=1}^t y_\alpha f_\alpha \\
	\text{subject to} & \quad p'(a \given x) \le y_0 + \sum_{\alpha=1}^t y_\alpha f_\alpha[p'] \quad \forall a\in\mathcal{A},\, \forall x\in\mathcal{X}_r,\,\forall p'\in\mathcal{Q} \,.
\end{split} \end{equation}
In \eqref{eq:noisy_conic_dual3}, we are thus solving a problem completely analogous to \eqref{eq:g_v(p)} except that the objective function now yields a bound on $G(\bvec f)$ that holds on the entire region $\freg$.
Since we minimize the objective function, we are searching for the best possible such bound.

Now, since the extreme points of $\freg$ are not necessarily quantum (i.e., do not necessarily belong to $\bvec{f}(\mathcal{Q})$), one could expect this upper bound to be strictly higher than the optimal quantum bound in $\freg$, contrarily to the primal formulation \eqref{eq:guess5}.
But this is not the case, as follows from the duality property of conic programs.
As we mentioned, the strong duality theorem for conic programs \cite{BV04} ensures that the value of the dual \eqref{eq:minmaxdual} matches that of the primal \eqref{eq:guess5}, as long as the primal has a strictly feasible point.
That is, a set of subnormalized probability vectors $\{\tilde p_{a,x}\}$ should exist such that all the equality constraints of the optimization are satisfied, and which lie in the interior of their respective cones, i.e., $\tilde p_{a,x} \in \interior(\tilde{\mathcal Q})$ for all $a \in \mathcal A, x \in \mathcal X_r$.
We argue that when the primal is not infeasible, this is almost always the case.

Consider the set $\bvec f[\interior(\mathcal Q)]$, that is, the set of Bell value vectors that are compatible with a non-extremal quantum behavior.
If the intersection of this set with the confidence region $\freg$ is non-null, we have strict feasibility of \eqref{eq:guess5}.
Indeed, let $p \in \interior(\mathcal Q)$ be a behavior such that $\bvec f[p] \in \freg$.
Then, any decomposition of $p$ as a sum of points in $\interior(\tilde{\mathcal Q})$, for instance $\tilde p_{a,x} = p / (\abs{\mathcal A} \times \abs{\mathcal X_r})$, gives a strictly feasible point for \eqref{eq:guess5}.

On the other hand, if the closure $\overline{\bvec f[\mathcal Q]}$ has no intersection with $\freg$, then the primal \eqref{eq:guess5} is infeasible and the dual \eqref{eq:minmaxdual} diverges to $-\infty$, as a rather straightforward application of the hyperplane separation theorem shows.

The last possibility is that $\bvec f[\interior(\mathcal Q)]$ has a point of tangency with $\freg$ or, in terms of separating hyperplanes, that there exist $\{y_0, y_\alpha\}$ such that $\sum_\alpha y_\alpha f_\alpha > y_0$ for all $\bvec f \in \bvec f[\interior(\mathcal Q)]$, and $\sum_\alpha y_\alpha f_\alpha \le y_0$ for all $\bvec f \in \freg$.
In this case, there might only exist non-strictly feasible points for the primal, while the dual does not diverge.
Without strong duality, there is no guarantee that the two resulting values are the same.
However, this case is irrelevant to us, as the chances that the confidence region $\freg$ around our estimated frequencies is tangent to $\bvec f[\interior(\mathcal Q)]$ are essentially zero.

Thus, in practice, the primal is either infeasible or it is strictly feasible.
Hence, when the dual \eqref{eq:minmaxdual} converges, we can safely conclude strong duality, and the optimum of the dual is indeed not worsened by allowing the dual to select $\bvec f \notin \bvec f[\mathcal Q]$, since the same value as computed by the primal \eqref{eq:guess5} is achieved by a feasible point $\{\tilde p_{a,x}\}$ such that $G(\freg) = G(\bvec f[p])$ with $p = \sum_{a,x} \tilde p_{a,x} \in \mathcal Q$.

\end{document}